\documentclass[12pt]{article}

\usepackage{amsmath, amssymb, amsthm}
\usepackage{amsfonts}
\usepackage{mathtools}        
\usepackage{bm}               

\usepackage[T1]{fontenc}
\usepackage[utf8]{inputenc}

\usepackage{graphicx}
\usepackage{booktabs}         
\usepackage{caption}
\usepackage{float}
\usepackage{tabularx}
\usepackage[hidelinks]{hyperref}
\usepackage{xurl}           

\usepackage[authoryear, round]{natbib}
\usepackage{setspace}
\def\spacingset#1{\renewcommand{\baselinestretch}{#1}\small\normalsize}

\usepackage{xcolor}
\usepackage{enumitem}         
\usepackage[nopatch=footnote]{microtype} 
\usepackage{etoolbox}
\usepackage{titlesec}
\usepackage{comment}
\usepackage[multiple]{footmisc}

\newtheorem{theorem}{Theorem}
\newtheorem{lemma}{Lemma}
\newtheorem{corollary}{Corollary}

\newtheorem{remark}{Remark}
\theoremstyle{definition}

\begin{document}

\spacingset{1}  

\vspace*{0.5in}  

\begin{center}
  {\LARGE\bf Generated outcomes as generated regressors: Equivalences in recursive causal estimation}
\end{center}
\medskip
\begin{center}
  Wisse Rutgers$^*$ \quad and \quad Rahul Singh$^\dagger$ \\[6pt]
  {\small $^*$CEMFI, \texttt{wisse.rutgers@cemfi.edu.es}} \\
  {\small $^\dagger$Harvard University, \texttt{rahul\_singh@fas.harvard.edu}} \\[10pt]
  June 2026
\end{center}

\bigskip

\begin{abstract}
\noindent Time-varying treatment effects, surrogate-identified treatment effects, and mediation effects can all be written as recursive regressions, in which each regression's predicted values become generated outcomes for the next regression. 
We study how standard causal estimators behave in this setting. 
Formally, we compare the recursive plug-in, recursive balancing weight, and recursive doubly robust estimators. 
When every stage is fitted by ordinary least squares (OLS), the three recursive estimators coincide in any finite sample, whether or not the models are correctly specified. 
As such, estimation by recursively regressing generated outcomes is numerically equivalent to estimation by recursively balancing generated regressors. 
Under ridge penalisation for the balancing weights, the doubly robust estimator is a backward recursion of stage-wise blends of penalised and OLS regressions. 
The weight on the recursive OLS regression decays geometrically in the number of time periods. 
Therefore, the intuition from the cross-sectional setting, where the bias  correction moves the estimator towards OLS, applies less and less as the number of time periods increases.
For general convex penalties, we derive an identity at each stage. 
\end{abstract}


\vfill
\newpage

\spacingset{1.5}


\section{Introduction}
\label{sec:intro}

Many of the counterfactual means estimated in empirical research are identified as recursions of regressions.
For example, the mean outcome under a two-period treatment path is identified by two regressions iterated back to back: first, the outcome is regressed on the full history;  then the predicted outcome, read at the counterfactual path, is regressed on the first period \citep{robins1986}. 
The mean outcome using surrogate variables shares this structure: the surrogate index is the regression of the long-run outcome on the surrogates in one sample, which is then regressed on the treatment and covariates in another sample \citep{prentice1989surrogate}.
The natural indirect effect in mediation analysis is yet another example: first, the outcome is regressed on treatment, mediator and covariates; then the predicted outcome, evaluated at the opposite treatment, is regressed on the treatment and covariates \citep{robins1992identifiability, pearl2001direct}. 
Across all of these examples, an initial regression generates a prediction, which is used as a generated outcome in a later regression. 
The classic identification results are stated as recursive regressions, i.e. with generated outcomes.

While the recursive regression formulation is in terms of the outcome mechanism, another formulation is in terms of the treatment mechanism. For example, the mean outcome under a two-period treatment path can also be expressed as the mean outcome after recursive reweighting by inverse propensity scores. These recursive weights may be viewed as recursively balancing the covariates of the treated and untreated subpopulations, using baseline covariates in the initial balancing and generated covariates in the subsequent balancing. Yet another formulation is doubly robust by augmenting the treatment mechanism formulation (in terms of generated regressors) with the outcome mechanism formulation (in terms of generated outcomes) \citep{bang2005doublyrobust}. Phrased another way, generated regressors can be used to debias generated outcomes. In summary, empirical researchers studying time-varying treatment effects, surrogate analysis, and mediation analysis face a plethora of estimators.

In cross sectional causal inference, it is well known that certain regression, balancing weight, and doubly robust estimators are numerically equivalent \citep{robins_performance}. Recent work has shown that, in cross sectional causal inference, debiasing a regularized regression estimator with a balancing weight effectively shrinks the regression towards least squares \citep{bruns2025augmented}. In this paper, we ask: does the equivalence documented for cross sectional causal inference extend to longitudinal causal inference, where there is the additional richness of generated outcomes and generated regressors? Moreover, does the interpretation of debiasing as shrinkage towards least squares generalize?

Our primary contribution is a positive answer to the first question: estimation by recursively regressing generated outcomes is numerically equivalent to estimation by recursively balancing generated regressors. 
This equivalence is exact when, in each period, the recursive regressions and recursive balancing weights are linear in the same dictionary of basis functions and unregularized. 
The equivalence is approximate when, in each period, the recursive regressions and recursive balancing weights are linear in the same dictionary of basis functions but regularized. 
Several nonparametric estimators are linear in their basis functions, such as series, lasso, and kernel methods; see Appendix~\ref{app:kernel-ridge}.

Our secondary contribution is a partially negative answer to the second question: the shrinkage interpretation of debiasing vanishes as the number of time periods grows.
While the doubly robust estimator with ridge-regularized balancing weights is still pulled towards OLS, the strength of that pull decays geometrically in the number of time periods. 
We conclude that the number of time periods is centrally important in understanding regularization. 
Only in the special case of cross sectional causal inference (with one time period) is it the case that debiasing and undersmoothing are two sides of the same coin.

\subsection{Related work}

Most directly, we build on previous work that characterizes numerical equivalences in cross-sectional causal inference \citep{rubin1980randomization,robins_performance, kline2011oaxaca, chattopadhyay2023implied}. See \citet{bruns2025augmented} for a comprehensive review and a general statement for the class of bounded linear functionals considered by \citet{chernozhukov2022automatic}. Recently, \citet{rotnitzky2025onestep} generalize this equivalence to the class of mixed bias linear functionals defined by \cite{rotnitzky2021mixedbias}. We extend the equivalence in a different direction, to the class of recursive functionals of  \cite{chernozhukov2022nested}. Unlike earlier work, the class we study includes the canonical models of time-varying treatment effects, surrogate analysis, and mediation analysis, which are widely used in empirical research.

Whereas previous works on recursive balancing weights propose estimators and analyse their (probabilistic) convergence properties, in this work we analyse their (deterministic) algebraic equivalences. 
Recursive balancing criteria have been proposed for parametric \citep{bang2005doublyrobust}, kernel \citep{kallus2021optimalbalancing}, high dimensional linear \citep{viviano2021dynamic}, and general machine learning \citep{chernozhukov2022nested} function spaces.
Our contribution is not to propose a new balancing criterion. Instead, we show when the balancing recursion with generated regressors is numerically identical to, or algebraically decomposable against, the regression recursion with generated outcomes.

More broadly, we contribute to a vast literature on doubly robust estimation, by shedding new light on how various estimators are interconnected. 
The literature establishes orthogonality, consistency, and asymptotic for recursive causal parameters \citep{molina2017multiple,luedtke2017sequential,rotnitzky2017multiply,chernozhukov2022nested}. 
We study a complementary question: how do the regression, balancing weight, and doubly robust estimators relate algebraically? We find that the main equivalence from cross sectional causal inference generalizes to longitudinal causal inference, however the intuition that debiasing amounts to under-smoothing does not.

Section~\ref{sec:setup} defines the class of recursive functionals, and details three leading examples: time-varying treatment, surrogate, and mediation analysis. 
Sections~\ref{sec:ols}, \ref{sec:ridge}, and~\ref{sec:general} study the cases with no regularization, ridge regularization, and general convex regularization, respectively. 
Section~\ref{sec:conclusion} concludes.


\section{Setup and framework}
\label{sec:setup}

\subsection{General recursive functionals}
\label{sec:nested-functionals}

We observe $n$ i.i.d.\ copies of a random vector $W$, the concatenation of all variables observed across time periods. The parameters we study are built from $T \geq 1$ stages of recursive conditional expectations, and we call $T$ the number of time periods. Following \citet{chernozhukov2022nested}, each stage $t$ contains three time-varying objects: (i) the conditioning variables $Z_t$; (ii) an expectation $\mathbb{E}_t(\cdot)$, taken with respect to a (possibly) time-varying measure; and (iii) a linear formula $m_t$, defined below. A base expectation $\mathbb{E}_0(\cdot)$ averages over the population that defines the target. In the time-varying treatment example below $Z_t$ records the full observed history and expectation measures do not change. However,  in the surrogacy example the treatment leaves the conditioning set between stages and the expectation measure changes from the observational group to the experimental group. Our setting nests both.  When every stage shares the population measure, we drop the subscripts and write $\mathbb{E}(\cdot)$. The last stage, $t = T$, is the one with the outcome $Y$ as its response; the first stage, $t = 1$, defines the target.

That target, $\theta_0$, is an average of a counterfactual outcome: the mean outcome had treatment followed a counterfactual path. Units do not necessarily follow the full path, so $\theta_0$ is a functional of a distribution we never fully observe. Time-varying treatment effects, surrogate-identified treatment effects, and mediation effects all take this form, and they are our running examples. Setting $T = 1$ collapses the framework to the single-period problem of \citet{chernozhukov2022automatic,bruns2025augmented}, which is our reference point throughout.

\subsubsection{Identification via backward outcome regressions}
\label{sec:identification}

The recursive structure identifies $\theta_0$ by a backward recursion of conditional expectations. It is the g-computation formula of \citet{robins1986} extended to additional counterfactuals.

The recursion starts at the last stage, $t = T$, which regresses the outcome on the variables observed there:
\begin{equation}
    f_T(z_T) := \mathbb{E}_T[Y \mid Z_T = z_T],
    \label{eq:fT}
\end{equation}
where $Z_T$ collects those final-stage variables. Each outer stage regresses the previous stage's output on its own conditioning set. For $t = T-1, \ldots, 1$,
\begin{equation}
    f_t(z_t) := \mathbb{E}_t\!\left[m_{t+1}(W;\, f_{t+1}) \mid Z_t = z_t\right],
    \label{eq:ft-recursive}
\end{equation}
where $m_{t+1}(W; g)$ is a mean-square-continuous linear functional of $g$. We call $m_{t+1}(W; f_{t+1})$ a generated outcome: it is built from nuisance $f_{t+1}$ and evaluated at an argument that may be counterfactual. The generated outcome then serves as the response one stage further out, so identification iterates backward through the stages.

Taking the first-stage expectation identifies the target by the mean of the first-stage generated outcome,
\begin{equation}
    \theta_0 = \mathbb{E}_0\!\left[m_1(W;\, f_1)\right].
    \label{eq:theta}
\end{equation}
Any parameter of the form~\eqref{eq:fT}--\eqref{eq:theta} is a recursive functional in the sense of \citet{chernozhukov2022nested}: a model in this class is specified by the number of time periods $T$ and the time-varying objects $\{Z_t, \mathbb{E}_t(\cdot), m_t\}$. The right-hand side of~\eqref{eq:theta} is a functional of the observed-data distributions, and it is the object every estimator in this paper targets. What gives the recursion a causal meaning is a separate matter, and the identifying assumptions differ across applications: time-varying treatment, surrogacy, and mediation each rest on their own conditions, stated with each running example in Section~\ref{sec:examples}. The algebra of this paper needs only the recursive form~\eqref{eq:fT}--\eqref{eq:theta}, whatever assumptions give it causal content.

\subsubsection{Identification via forward Riesz regressions}
\label{sec:riesz-identification}

Each outcome regression has a forward-pass Riesz representer $\alpha_t$, defined stage by stage through the Riesz representation theorem \citep{chernozhukov2022nested}. The Riesz representer is a function, and its evaluations give a vector called the balancing weights.  We take $\alpha_t$ to be the representer in the space of functions of $Z_t$ that are square-integrable under the stage-$t$ measure, so it is $Z_t$-measurable by construction, and we write $\alpha_t(Z_t)$ throughout.\footnote{If a representer on the full vector $W$ exists, $\alpha_t(Z_t)$ is its conditional expectation given $Z_t$ under the stage-$t$ measure, the minimal such representer; this convention is what makes the linear form $\alpha_t(Z_t) = \phi_t'\eta_t$ of Section~\ref{sec:linear-setup} well-posed, since any representer on $W$ is pinned down only up to additions with zero conditional mean given $Z_t$.} Set $\alpha_0 := 1$. The first representer $\alpha_1$ turns the functional $m_1$ into an inner product:
\begin{equation}
    \mathbb{E}_0[m_1(W; g)] = \mathbb{E}_1[\alpha_1(Z_1)\, g(Z_1)]
    \quad \text{for all square-integrable } g,
    \label{eq:alpha1-riesz}
\end{equation}
and for $t = 2, \ldots, T$, $\alpha_t$ represents the functional $L_t(g) := \mathbb{E}_{t-1}[\alpha_{t-1}(Z_{t-1})\, m_t(W; g)]$:
\begin{equation}
    \mathbb{E}_{t-1}[\alpha_{t-1}(Z_{t-1})\, m_t(W; g)] = \mathbb{E}_t[\alpha_t(Z_t)\, g(Z_t)]
    \quad \text{for all square-integrable } g.
    \label{eq:alphat-riesz}
\end{equation}
Existence of these representers is at the population level: each represented functional must be mean-square continuous, which in the running examples follows from sequential positivity and overlap between adjacent stage measures. The finite-sample identities of Sections~\ref{sec:ols}--\ref{sec:general} do not need it: they are statements about the coefficient vectors the stage programmes return, and they hold whether or not a population $\alpha_t$ exists. Once $\alpha_{t-1}$ is replaced by an estimate $\hat\alpha_{t-1}$, that estimate enters the stage-$t$ Riesz regression as a generated regressor, the forward-pass counterpart of the generated outcome.

The representers identify $\theta_0$ by a forward recursion. Start from $\theta_0 = \mathbb{E}_0[m_1(W; f_1)] = \mathbb{E}_1[\alpha_1(Z_1) f_1(Z_1)]$ and substitute $f_1(Z_1) = \mathbb{E}_1[m_2(W; f_2) \mid Z_1]$:
\[
    \theta_0 = \mathbb{E}_1[\alpha_1(Z_1)\, m_2(W; f_2)] = \mathbb{E}_2[\alpha_2(Z_2)\, f_2(Z_2)],
\]
where the second equality uses~\eqref{eq:alphat-riesz}. Iterating through all $T$ stages gives the inverse-probability-weighted (IPW) representation
\begin{equation}
    \theta_0 = \mathbb{E}_T[\alpha_T(Z_T)\, Y],
    \label{eq:ipw-id}
\end{equation}
a weighted average of the outcome under the last-stage measure. In this final expression, $\alpha_T(Z_T)$ may be viewed as the balancing weights for generated regressors. Specifically, the regressors being balanced at time $T$ are recursively generated by $\alpha_{t}(Z_t)$ for $t<T$.

\subsubsection{The doubly robust moment}
\label{sec:debiasing}

Estimating $f_t$ can introduce regularisation bias. In the plug-in estimator $\hat{\mathbb{E}}_0[m_1(W; \hat{f}_1)]$, that bias accumulates from the last stage backward. First-order bias can be eliminated by a Neyman-orthogonal moment, built by applying the Riesz representation theorem recursively.

Write the stage-$t$ residual, i.e. the prediction error of $f_t$ for its generated outcome, as
\begin{equation}
    \varepsilon_t(W) :=
    \begin{cases}
        m_{t+1}(W; f_{t+1}) - f_t(Z_t) & t = 1, \ldots, T-1, \\[3pt]
        Y - f_T(Z_T)                    & t = T.
    \end{cases}
    \label{eq:residuals}
\end{equation}
We write $\hat\varepsilon_t$ for the same residual evaluated at estimated nuisances, with $\hat{f}_{t+1}$ and $\hat{f}_t$ in place of $f_{t+1}$ and $f_t$. Collecting the corrections gives the orthogonal moment \citep{chernozhukov2022nested}:
\begin{equation}
    \theta_0 = \mathbb{E}_0\!\left[m_1(W; f_1)\right] + \sum_{t=1}^{T} \mathbb{E}_t\!\left[\alpha_t(Z_t)\, \varepsilon_t(W)\right].
    \label{eq:debiased-general}
\end{equation}
Each correction $\mathbb{E}_t[\alpha_t(Z_t)\, \varepsilon_t(W)]$ is taken under its own stage measure and is zero at the true nuisances: $f_t$ is by definition the stage-$t$ conditional mean of its response, so $\mathbb{E}_t[\varepsilon_t(W) \mid Z_t] = 0$, and the weight $\alpha_t(Z_t)$ is a function of the same conditioning set. Orthogonality follows from the bias expansion below, in which every term is a product of two estimation errors. The estimator built from the orthogonal moment attains the semiparametric efficiency bound in many cases.

The recursive moment has a recursive mixed bias property \citep{chernozhukov2022nested}, the recursive counterpart of the bias structure that \citet{rotnitzky2021mixedbias} characterise for a single pair of nuisances. Take any alternative nuisances $(f_t^*, \alpha_t^*)$ with errors $\tilde{f}_t = f_t^* - f_t$ and $\tilde\alpha_t = \alpha_t^* - \alpha_t$. Evaluating the moment at the alternatives, the bias is
\begin{equation}
    \theta^* - \theta_0 = \sum_{t=1}^{T} \mathbb{E}_t\!\left[\tilde\alpha_t(Z_t)\bigl(m_{t+1}(W; \tilde{f}_{t+1}) - \tilde{f}_t(Z_t)\bigr)\right],
    \label{eq:mixed-bias}
\end{equation}
with the convention $m_{T+1}(W; \tilde{f}_{T+1}) := 0$. The inner error enters as the generated outcome $m_{t+1}(W; \tilde{f}_{t+1})$. Each summand is a product of a representer error and an outcome error, so it vanishes when $\tilde\alpha_t = 0$ (the stage-$t$ representer is correct) or when $\tilde{f}_{t+1} = \tilde{f}_t = 0$ (both adjacent outcome regressions are correct). This stage-wise robustness is the recursive analogue of double robustness: at $T = 1$, equation~\eqref{eq:mixed-bias} is the single term $-\mathbb{E}_1[\tilde\alpha_1 \tilde{f}_1]$, which vanishes if either $\alpha_1$ or $f_1$ is correct. \citet{rotnitzky2017multiply} establish multiple robustness of this stage-wise kind for estimators of the time-varying treatment estimand of Example~2.

Two special cases of the moment~\eqref{eq:debiased-general} are given below. At $T = 1$, with a single population measure, it is the standard doubly robust (AIPW) moment
\begin{equation}
    \theta_0 = \mathbb{E}[m_1(W; f_1)] + \mathbb{E}[\alpha_1(Z_1)(Y - f_1(Z_1))],
    \label{eq:dr-t1}
\end{equation}
our non-recursive baseline. With distinct base and stage measures it is instead the covariate-shift moment of \citet{chernozhukov2023covshifts}. At $T = 2$ it expands to
\begin{align}
    \theta_0 = \;
    &\mathbb{E}_0[m_1(W; f_1)] \notag \\
    &+ \mathbb{E}_1\!\left[\alpha_1(Z_1)\bigl(m_2(W; f_2) - f_1(Z_1)\bigr)\right] \notag \\
    &+ \mathbb{E}_2\!\left[\alpha_2(Z_2)\bigl(Y - f_2(Z_2)\bigr)\right],
    \label{eq:debiased-moment}
\end{align}
the leading recursive case for the running examples.

\subsection{Examples}
\label{sec:examples}

We provide four examples that are nested by the recursive framework. The first ($T = 1$) is the non-recursive baseline. The other three ($T = 2$) show how recursion arises in time-varying treatment, surrogacy, and mediation.

\paragraph{Example 1: Single period ($T = 1$).}
We observe $n$ i.i.d.\ copies of $W = (Z, Y)$, where $Z$ collects covariates and possibly a treatment indicator $D$. The single nuisance function is $f_1(z) = \mathbb{E}[Y \mid Z = z]$, the single Riesz representer $\alpha_1$ satisfies~\eqref{eq:alpha1-riesz}, and $\theta_0 = \mathbb{E}[m_1(W; f_1)]$. The doubly robust moment reduces to the standard AIPW form~\eqref{eq:dr-t1}. This setting covers the average treatment effect \citep{robins1994estimating} and other bounded linear functionals \citep{chernozhukov2022locallyrobust}. \citet{bruns2025augmented} show that when both the outcome model and the Riesz representer are estimated by ridge regression, the augmented estimator is algebraically equivalent to a single undersmoothed outcome regression; this equivalence is the direct precursor to our recursive results.

\paragraph{Example 2: $T = 2$ dynamic treatment effects.}
We observe $n$ i.i.d.\ trajectories $W = (X_1, D_1, X_2, D_2, Y)$, with time-varying covariates $X_t$, binary treatments $D_t$, and final outcome $Y$. We set $Z_1 = (X_1, D_1)$ and $Z_2 = (X_1, D_1, X_2, D_2)$. The parameter of interest is the counterfactual mean $\theta_0(d_1, d_2) := \mathbb{E}[Y(d_1, d_2)]$. Three  conditions at each stage $t$ identify $\theta_0$ via the g-computation formula \citep{robins1986}. (i) If $(D_1,D_2)=(d_1,d_2)$ then $Y=Y(d_1,d_2)$. (ii) Potential outcomes are conditionally independent of the stage-$t$ treatment.  (iii) Treatment propensities are bounded away from zero. Under these conditions the nuisance functions are:
\begin{align*}
    f_2(Z_2) &= \mathbb{E}[Y \mid X_1, D_1, X_2, D_2], \\
    f_1(Z_1) &= \mathbb{E}[f_2(X_1, d_1, X_2, d_2) \mid X_1, D_1],
\end{align*}
where $m_2(W; g) = g(X_1, d_1, X_2, d_2)$ evaluates $g$ at the counterfactual treatment sequence $(d_1, d_2)$, and $\theta_0 = \mathbb{E}[f_1(X_1, d_1)]$, so $m_1(W; g) = g(X_1, d_1)$. 

\paragraph{Example 3: Surrogacy.}
We observe an experimental sample $(X, D, S)$ and an observational sample $(X, S, Y)$, with binary treatment $D$, surrogate $S$, and long-run outcome $Y$. The estimand is $\theta_0 = \mathbb{E}_0[Y(1)] - \mathbb{E}_0[Y(0)]$. Identification rests on the surrogate validity assumption of \citet{athey2025surrogate}, which requires the conditional outcome mean to agree across the two populations, $\mathbb{E}_2[Y \mid X, S] = \mathbb{E}_1[Y \mid X, S]$, where the measures $\mathbb{E}_1$ and $\mathbb{E}_2$ in this case denote the experimental and observational samples, and $\mathbb{E}_0$ is the target population expectation. The nuisances are then
\begin{align*}
    f_2(X, S) &= \mathbb{E}_2[Y \mid X, S] , \\
    f_1(X, D) &= \mathbb{E}_1[f_2(X, S) \mid X, D] .
\end{align*}
The stage functionals are evaluations: $m_2(W; g) = g(X, S)$ reads $g$ at the realised surrogate, with no counterfactual substitution, so the recursion content of the inner stage is the shift across samples and not a shift of the treatment. The outer functional is a contrast, $m_1(W; g) = g(X, 1) - g(X, 0)$. The two-sample structure means the inner and outer regressions are estimated on different datasets.

Two additional assumptions identify $\theta_0$. Common covariates, $p_1(X) = p_2(X)$, allows the surrogate distribution to be shifted across samples. Experimental unconfoundedness, $S(d) \perp D \mid X$ under measure $\mathbb{E}_1$, identifies the surrogate distribution under each arm. 

\paragraph{Example 4: Mediation analysis.}
We observe $n$ i.i.d.\ copies of $W = (X, D, M, Y)$, with covariates $X$, binary treatment $D$, mediator $M$, and outcome $Y$. The estimand is the cross-world mean $\theta_0 = \mathbb{E}[Y(d,\, M(1{-}d))]$, the mean outcome when the treatment is set to $d$ but the mediator is drawn as under treatment $1{-}d$  \citep{robins1992identifiability, pearl2001direct}. Under a sequential ignorability assumption \citep{imai2010mediation}, $\theta_0$ is identified by the mediation formula, a $T = 2$ recursive functional:
\begin{align*}
    f_2(X, D, M) &= \mathbb{E}[Y \mid X, D, M], \\
    f_1(X, D)    &= \mathbb{E}[f_2(X, d, M) \mid X, D],
\end{align*}
where $m_2(W; g) = g(X, d, M)$ evaluates $g$ at treatment $d$, and $\theta_0 = \mathbb{E}[f_1(X, 1{-}d)]$, so $m_1(W; g) = g(X, 1{-}d)$. The same formula arises in the alternative framework of interventional counterfactuals \citep{robins2010alternative}. The two stages evaluate at different treatments. This sets mediation apart from Example~2, where one counterfactual sequence $(d_1, d_2)$ enters at both stages. 

\subsection{Linear estimation framework}
\label{sec:linear-setup}

To make the algebra explicit, similar to \citet{bruns2025augmented}, we work with estimators that are linear in a dictionary of basis functions. At each stage $t = 1, \ldots, T$, let $\phi_t : \mathcal{Z}_t \to \mathbb{R}^{k_t}$ be such a dictionary with each component square-integrable. Further, define $\phi_t^d$ as the image of the dictionary under the stage functional: $(\phi_t^d)_j := m_t(W;\, \phi_{t,j})$ for $j = 1, \ldots, k_t$. We take both nuisances to be linear in the dictionary, $f_t(Z_t) = \phi_t' \beta_t$ and $\alpha_t(Z_t) = \phi_t' \eta_t$, with coefficient vectors $\beta_t, \eta_t \in \mathbb{R}^{k_t}$. Since $m_t$ is linear in its functional argument, the definition gives the representation
\begin{equation}\label{eq:representability}
    m_t(W;\, \phi_t' b) = (\phi_t^d)'\, b
    \quad \text{for all } b \in \mathbb{R}^{k_t}, \quad t = 1, \ldots, T,
\end{equation}
so applying $m_t$ to a linear function swaps $\phi_t$ for $\phi_t^d$ and leaves the coefficients alone. The superscript $d$ records the leading case, in which $m_t$ substitutes a counterfactual treatment (or treatment path) and $\phi_t^d$ is simply $\phi_t$ with the treatment overwritten. 

When the estimand is a contrast, as in the surrogacy estimand $\mathbb{E}_0[Y(1)] - \mathbb{E}_0[Y(0)]$, the outer functional is the combination $m_1(W; g) = g(X, 1) - g(X, 0)$, so $\phi_1^d = \phi_1^1 - \phi_1^0$, a difference of evaluations instead of a single one. The surrogacy inner stage does not substitute anything. Instead, $m_2(W; g) = g(X, S)$ reads $g$ at the realised surrogate, so $\phi_2^d = \phi_2$, and this stage shifts from the observational sample to the experimental sample.

Many nonparametric estimators have the property of linearity in $\phi_t$, e.g. series, lasso, and kernel methods, as well as certain neural networks and random forests. See \citet[Section~2.3]{bruns2025augmented} for discussion.

We use the sample analogue of the time-varying $\mathbb{E}_t$ of Section~\ref{sec:nested-functionals}, denoted by $\hat{\mathbb{E}}_t$, which averages over the sample drawn from the stage-$t$ population. When there is only one measure, the single operator $\hat{\mathbb{E}}$ is used. Write $\hat{G}_t := \hat{\mathbb{E}}_t[\phi_t \phi_t']$ for the sample Gram matrix at stage $t$, and $\hat{M}_t := \hat{\mathbb{E}}_t[\phi_t\, (\phi_{t+1}^d)']$, $t = 1, \ldots, T-1$, for the cross-moment matrix linking adjacent stages, both under the outer stage-$t$ measure. Finally, at each stage $t$, let $P_t : \mathbb{R}^{k_t} \to \mathbb{R}$ be a convex penalty on the outcome coefficients, with parameter $\lambda_t \geq 0$, and $Q_t : \mathbb{R}^{k_t} \to \mathbb{R}$ a convex penalty on the Riesz coefficients, with parameter $\delta_t \geq 0$. In the notation of \citet{bruns2025augmented}, $\hat{G}_1$ is their $\frac{1}{n}\Phi_p^\top\Phi_p$ and $\hat{\mathbb{E}}_0[\phi_1^d]$ their $\bar\Phi_q$.

\subsubsection{Recursive outcome regression}
\label{sec:nested-regression}

The simplest estimator of $\theta_0$ is the plug-in estimator, which estimates each $f_t$ by an outcome regression and substitutes it into the identification formula. Each stage solves a regularised least-squares problem. The base case $t = T$ regresses $Y$ on the dictionary $\phi_T$:
\begin{equation}\label{eq:outcome-reg-T}
    \hat\beta_T = \arg\min_{\beta \in \mathbb{R}^{k_T}} \left\{
        \frac{1}{2}\hat{\mathbb{E}}_T\!\left[(Y - \phi_T'\beta)^2\right] + \tfrac{\lambda_T}{2}P_T(\beta)
    \right\},
\end{equation}
where $P_T$ is convex and $\lambda_T \geq 0$ the regularisation parameter. At each stage $t = T{-}1, \ldots, 1$, the inner coefficient $\hat\beta_{t+1}$ produces the generated outcome $(\phi_{t+1}^d)'\hat\beta_{t+1}$, which is the target for stage $t$:
\begin{equation}\label{eq:outcome-reg-t}
    \hat\beta_t = \arg\min_{\beta \in \mathbb{R}^{k_t}} \left\{
        \frac{1}{2}\hat{\mathbb{E}}_t\!\left[\bigl((\phi_{t+1}^d)'\hat\beta_{t+1} - \phi_t'\beta\bigr)^2\right] + \tfrac{\lambda_t}{2}P_t(\beta)
    \right\}.
\end{equation}
The first-order conditions of~\eqref{eq:outcome-reg-T}--\eqref{eq:outcome-reg-t} define the backward pass:
\begin{align}
    \hat{G}_T\hat\beta_T + \tfrac{\lambda_T}{2}\nabla P_T(\hat\beta_T) &= \hat{\mathbb{E}}_T[Y \phi_T], \notag \\
    \hat{G}_t\hat\beta_t + \tfrac{\lambda_t}{2}\nabla P_t(\hat\beta_t) &= \hat{M}_t\,\hat\beta_{t+1}, \quad t = T-1, \ldots, 1.
    \label{eq:beta-backward}
\end{align}
The plug-in estimator is then
\[
    \hat\theta^{P} := \hat{\mathbb{E}}_0\!\left[m_1(W;\, \hat{f}_1)\right]
    = \hat{\mathbb{E}}_0\!\left[(\phi_1^d)'\right]\hat\beta_1.
\]
Here, the superscript $P$ records the outcome-regression penalty path $P = (P_1, \ldots, P_T)$ that produces $\hat\beta_1$.

\subsubsection{Recursive Riesz regression}
\label{sec:nested-riesz}

The dual strategy estimates the representers $\alpha_t$ and reweights with them, instead of modelling the outcome functions. Each stage solves a regularised Riesz loss. The base case is stage $1$, which is a non-recursive problem expressed by \citet{chernozhukov2022automatic,singh2021kernel,chernozhukov2021riesz,chernozhukov2023covshifts} as:
\begin{equation}\label{eq:riesz-t1}
    \hat\eta_1 = \arg\min_{\eta \in \mathbb{R}^{k_1}} \left\{
        \frac{1}{2}\hat{\mathbb{E}}_1\!\left[(\phi_1'\eta)^2\right]
        - \hat{\mathbb{E}}_0\!\left[(\phi_1^d)'\eta\right]
        + \tfrac{\delta_1}{2}Q_1(\eta)
    \right\},
\end{equation}
where $Q_1$ is the convex penalty and $\delta_1 \geq 0$ the regularisation parameter. At each stage $t = 2, \ldots, T$, the inner coefficient $\hat\eta_{t-1}$ produces the generated regressor $\hat\alpha_{t-1} = \phi_{t-1}'\hat\eta_{t-1}$, which enters the stage-$t$ Riesz loss for the functional $L_t(g) = \hat{\mathbb{E}}_{t-1}[\hat\alpha_{t-1}(Z_{t-1})\, m_t(W; g)]$ of \citet{chernozhukov2022nested}:
\begin{equation}\label{eq:riesz-t}
    \hat\eta_t = \arg\min_{\eta \in \mathbb{R}^{k_t}} \left\{
        \frac{1}{2}\hat{\mathbb{E}}_t\!\left[(\phi_t'\eta)^2\right]
        - \hat{\mathbb{E}}_{t-1}\!\left[\phi_{t-1}'\hat\eta_{t-1}\, (\phi_t^d)'\eta\right]
        + \tfrac{\delta_t}{2}Q_t(\eta)
    \right\}.
\end{equation}
The first-order conditions of~\eqref{eq:riesz-t1}--\eqref{eq:riesz-t} define the forward pass:
\begin{align}
    \hat{G}_1\hat\eta_1 + \tfrac{\delta_1}{2}\nabla Q_1(\hat\eta_1) &= \hat{\mathbb{E}}_0[\phi_1^d], \notag \\
    \hat{G}_t\hat\eta_t + \tfrac{\delta_t}{2}\nabla Q_t(\hat\eta_t) &= \hat{M}_{t-1}'\,\hat\eta_{t-1}, \quad t = 2, \ldots, T.
    \label{eq:eta-forward}
\end{align}
The balancing weight estimator is then
\[
    \hat\theta^{Q} := \hat{\mathbb{E}}_T\!\left[\hat\alpha_T \cdot Y\right]
    = \hat{\mathbb{E}}_T\!\left[Y\phi_T'\right]\hat\eta_T,
\]
the inverse-probability-weighted form of the Riesz regression estimator. Here, the superscript $Q$ denotes the Riesz penalty path $Q = (Q_1, \ldots, Q_T)$ that produces $\hat\eta_T$. 

\subsubsection{Balancing weights and the recursive Riesz regression}
\label{sec:dcb}

The recursive Riesz regression of the previous subsection has an equivalent reading as a recursive balancing procedure, in which $\hat\alpha_t = \phi_t'\hat\eta_t$ assigns unit $i$ at stage $t$ the weight $w_{i,t} = \phi_t(z_{t,i})'\hat\eta_t$. The reweighted stage-$t$ feature profile is then $\tfrac{1}{n}\sum_{i=1}^n w_{i,t}\,\phi_t(z_{t,i}) = \hat{G}_t\hat\eta_t$, and the forward-pass problems that produce $\hat\eta_t$ trade the variance of these weights against their imbalance relative to a target profile. This is the structure that \citet{bruns2025augmented}, building on \citet{zubizarreta2015stable}, \citet{athey2018residual}, and \citet{hirshberg2021augmented}, and many others, analyse for $T=1$. The change from one step to multiple steps is that the non-recursive target $\bar\phi^d = \hat{\mathbb{E}}[\phi^d]$ is replaced by a generated target $\hat\tau_t$.

\paragraph{Equivalence of recursive balancing and recursive Riesz regression.} The $T=1$ balancing-Riesz equivalence carries over to the recursive setting stage by stage. At each stage $t$, define the recursive target vector
\[
    \hat\tau_t :=
    \begin{cases}
        \hat{\mathbb{E}}_0[\phi_1^d] & t = 1, \\
        \hat{M}_{t-1}'\hat\eta_{t-1} & t = 2, \ldots, T,
    \end{cases}
\]
which replaces the simple target $\bar\phi^d$ of the non-recursive case. 

\begin{lemma}[Recursive balancing-Riesz equivalence]\label{prop:nested-equiv}
Fix a stage $t \in \{1, \ldots, T\}$ and condition on the previous-stage estimate $\hat\eta_{t-1}$ (and hence $\hat\tau_t$). Assume $\hat{G}_t$ is invertible, and pair the imbalance norm with its dual penalty: $\|\cdot\|_* = \|\cdot\|_2$ with $Q_t(\eta) = \|\eta\|_2^2$, or $\|\cdot\|_* = \|\cdot\|_\infty$ with $Q_t(\eta) = \|\eta\|_1$.\footnote{Without invertibility the constrained and penalised forms determine $\hat\eta_t$ only up to an element of the null space of $\hat{G}_t$. Any such element $v$ satisfies $\phi_t'v = 0$ at every observation in the stage-$t$ sample, so all solutions give the same in-sample weights and the same next-stage target $\hat\tau_{t+1}$, and the estimator is unchanged as long as the weights are evaluated on the sample that defines $\hat{G}_t$.} Then, for the set of unique solutions, the following three formulations yield the same solution $\hat\eta_t$, up to a one-to-one, data-dependent mapping between the regularisation parameters $\delta_t^{(C)} > 0$, $\delta_t^{(P)}$, and $\delta_t^{(R)}$:
\begin{enumerate}[label=(\roman*)]
    \item \emph{Constrained (minimum variance) form:}
    \begin{equation}\label{eq:nested-constrained}
        \hat\eta_t = \arg\min_{\eta} \;\frac{1}{2}\eta'\hat{G}_t\eta
        \quad \text{subject to} \quad
        \left\|\hat{G}_t\eta - \hat\tau_t\right\|_* \leq \delta_t^{(C)}.
    \end{equation}
    \item \emph{Penalised imbalance form:}
    \begin{equation}\label{eq:nested-penalised}
        \hat\eta_t = \arg\min_{\eta} \;\left\|\hat{G}_t\eta - \hat\tau_t\right\|_*^2
        + \delta_t^{(P)}\,\eta'\hat{G}_t\eta.
    \end{equation}
    \item \emph{Riesz regression form:}
    \begin{equation}\label{eq:nested-riesz-form}
        \hat\eta_t = \arg\min_{\eta} \;\frac{1}{2}\eta'\hat{G}_t\eta
        - \hat\tau_t'\eta + \tfrac{\delta_t^{(R)}}{2} Q_t(\eta).
    \end{equation}
\end{enumerate}
\end{lemma}

\begin{proof}
    See Appendix~\ref{sec:setup-proof}
\end{proof}

Stage by stage, then, the forward-pass recursive Riesz regression can be viewed as a sequential balancing procedure. The constrained, penalised, and Riesz formulations of $\hat\eta_t$ all return the same weights.

\paragraph{Connection to dynamic covariate balancing.} \citet{viviano2021dynamic} propose dynamic covariate balancing (DCB) for the time-varying treatment setting where at each period $t$, they solve a constrained programme minimising the $\ell_2$ norm of the weights subject to the dynamic balance condition
\begin{equation}\label{eq:dcb-balance}
    \Bigl\|\textstyle\sum_{i=1}^n \hat\gamma_{i,t-1}\, H_{i,t} - \sum_{i=1}^n \gamma_{i,t}\, H_{i,t}\Bigr\|_\infty \leq \kappa_t,
\end{equation}
where $H_{i,t}$ is the observed history of unit $i$ at period $t$, $\gamma_{i,t}$ is the period-$t$ weight on unit $i$ being chosen, $\hat\gamma_{i,t-1}$ is the weight estimated at period $t-1$, the tolerance $\kappa_t$ bounds the imbalance, and $\hat\gamma_{i,0} = 1/n$ initialises the recursion. DCB also restricts each weight vector to the simplex (non-negative weights summing to one), caps the largest weight, and restricts the weights to the units on the target path, $\gamma_{i,t} = 0$ unless $D_{i,1:t} = d_{1:t}$. On the retained units the observed history equals the counterfactual one, which is what lets the same $H_{i,t}$ stand in for $\phi_t^d$ on both sides of~\eqref{eq:dcb-balance}.

Under the identification $\hat\alpha_t(W_i) = n\,\gamma_{i,t}$ (the DCB weights sum to one while the representer averages to one), condition~\eqref{eq:dcb-balance} bounds the same imbalance as~\eqref{eq:nested-constrained} with $\|\cdot\|_* = \|\cdot\|_\infty$ and $\delta_t^{(C)} = \kappa_t$. The two programmes minimise the weight variance over different sets, form~\eqref{eq:nested-constrained} over the span of the dictionary on the full sample and DCB over free weight vectors that are zero off the target path, so for a generic dictionary the solutions differ even at exact balance. Without the simplex constraint and the weight cap, the two coincide when the stage-$t$ dictionary is saturated in the treatment path, meaning every entry either is $\mathbf{1}\{D_{1:t} = d_{1:t}\}$ times a function of the covariate history or vanishes on the path. 

\subsubsection{The doubly robust estimator}
\label{sec:dml}

The doubly robust estimator takes the sample analogue of the orthogonal moment~\eqref{eq:debiased-general}, evaluated at the estimated nuisances:
\begin{equation}
    \hat\theta^{DR} := \hat{\mathbb{E}}_0\!\left[m_1(W;\, \hat{f}_1)\right]
    + \sum_{t=1}^{T} \hat{\mathbb{E}}_t\!\left[\hat\alpha_t\, \hat\varepsilon_t\right],
    \label{eq:dr-estimator}
\end{equation}
where $\hat\varepsilon_t$ are the sample residuals of Section~\ref{sec:debiasing}: $\hat\varepsilon_t = (\phi_{t+1}^d)'\hat\beta_{t+1} - \phi_t'\hat\beta_t$ for $t < T$, by~\eqref{eq:representability}, and $\hat\varepsilon_T = Y - \phi_T'\hat\beta_T$. In the linear framework, $\hat\theta^{DR}$ is estimated using the backward-pass coefficients $\hat\beta_t$ from~\eqref{eq:outcome-reg-T}--\eqref{eq:outcome-reg-t} and the forward-pass coefficients $\hat\eta_t$ from~\eqref{eq:riesz-t1}--\eqref{eq:riesz-t}.


\section{No regularisation: OLS}
\label{sec:ols}

When we estimate all $2T$ nuisance functions by OLS, with no penalty at any stage, the recursive plug-in estimator $\hat\theta^{P}$, the recursive balancing weight estimator $\hat\theta^{Q}$, and the doubly robust estimator $\hat\theta^{DR}$ coincide exactly. This is an algebraic identity, so it holds in any sample in which each Gram matrix $\hat{G}_t$ is invertible, and whether or not the linear models are correctly specified.

Throughout this section each Gram matrix $\hat{G}_t$ is invertible, that is, the stage-$t$ dictionary has full column rank in the sample, which requires $k_t \leq n_t$. OLS falls under our linear estimators where the regularisation penalties are set to zero, ($P_t = Q_t = 0$ for all $t$), such that the backward and forward passes~\eqref{eq:beta-backward}--\eqref{eq:eta-forward} lead to FOCs:
\begin{align}
    \hat\beta_T^{OLS} &= \hat{G}_T^{-1}\,\hat{\mathbb{E}}_T[Y\phi_T], \notag \\
    \hat\beta_t^{OLS} &= \hat{G}_t^{-1}\,\hat{M}_t\,\hat\beta_{t+1}^{OLS}, \quad t = T-1,\ldots,1,
    \label{eq:backward-ols}
\end{align}
and
\begin{align}
    \hat\eta_1^{OLS} &= \hat{G}_1^{-1}\,\hat{\mathbb{E}}_0[\phi_1^d], \notag \\
    \hat\eta_t^{OLS} &= \hat{G}_t^{-1}\,\hat{M}_{t-1}'\,\hat\eta_{t-1}^{OLS}, \quad t = 2,\ldots,T.
    \label{eq:forward-ols}
\end{align}
The recursive plug-in and balancing weight estimators of Section~\ref{sec:linear-setup} specialise to:
\begin{equation}\label{eq:theta-reg-ipw}
    \hat\theta^{P} = \hat{\mathbb{E}}_0\!\left[(\phi_1^d)'\right]\hat\beta_1^{OLS}, \qquad
    \hat\theta^{Q} = \hat{\mathbb{E}}_T\!\left[Y\phi_T'\right]\hat\eta_T^{OLS}.
\end{equation}
The next proposition records the equivalence and the two facts behind it: every correction is individually zero, and every stage-wise bilinear form already equals the target.
\begin{theorem}[OLS equivalence]\label{prop:ols-equiv}
Assume each Gram matrix $\hat{G}_t$ is invertible, $t = 1, \ldots, T$. Under OLS estimation of all nuisance functions:
\begin{enumerate}[label=(\roman*)]
    \item Every debiasing correction vanishes:
    \begin{equation}\label{eq:ols-corrections-zero}
        \hat{\mathbb{E}}_t\!\left[\hat\alpha_t^{OLS}\,\hat\varepsilon_t\right] = 0, \quad t = 1,\ldots,T.
    \end{equation}
    \item All intermediate terms are equal:
    \begin{equation}\label{eq:ols-bilinear}
        (\hat\eta_t^{OLS})'\hat{G}_t\,\hat\beta_t^{OLS} = \hat\theta^{P} = \hat\theta^{Q}, \quad t = 1,\ldots,T,
    \end{equation}
    and the cross-stage inner products satisfy $(\hat\eta_t^{OLS})'\hat{M}_t\,\hat\beta_{t+1}^{OLS} = \hat\theta^{P}$ for $t = 1,\ldots,T{-}1$.
\end{enumerate}
Consequently, $\hat\theta^{DR} = \hat\theta^{P} = \hat\theta^{Q}$ in any sample in which each $\hat{G}_t$ is invertible, for any number of time periods $T$, whether or not the linear models are correctly specified.
\end{theorem}

\begin{proof}
See Appendix~\ref{app:proof-ols}.
\end{proof}

Part~(i) does not use OLS on the Riesz side. The proof uses only the outcome-side normal equations and the linearity of the weights, so the corrections vanish for any weights in the span of the regressors. This is a doubly robust property of OLS: adding any linear balancing weights to an OLS outcome model leaves the estimator unchanged. Equivalently, one can show that adding a linear outcome regression to OLS balancing weights does not change the estimator.

\section{Ridge regularisation}
\label{sec:ridge}

With ridge Riesz regressions and any outcome model, the doubly robust estimator can be represented as a single recursive outcome regression, but with augmented coefficients set by a backward recursion. Each stage shrinks towards OLS, but towards an OLS prediction of an already-augmented stage. The weight the estimator leaves on the pure OLS plug-in estimator decays geometrically in the number of time periods. 

Each stage contains two penalties: $\lambda_t$ on the outcome regression and $\delta_t$ on the Riesz regression, each a positive definite diagonal matrix. Define
\[
    A_t := (\hat{G}_t + \delta_t)^{-1}\hat{G}_t, \qquad
    I - A_t = (\hat{G}_t + \delta_t)^{-1}\delta_t,
\]
where the second equation follows from $(\hat{G}_t + \delta_t) - \hat{G}_t = \delta_t$. The matrix $A_t$ is the ridge shrinkage operator of the time $t$ Riesz regression. For a stage-$t$ least-squares problem with a fixed target $b$, ridge predicts $\hat\beta^{R} = A_t \hat\beta^{OLS}$, so ridge is $A_t$ applied to the OLS coefficient, and $I - A_t$ measures the shrinkage between them. 

For a scalar penalty $\delta_t = \delta I$ the two matrices commute, $A_t$ is symmetric, and $0 \prec A_t \prec I$. For a general diagonal $\delta_t$ the matrix $A_t$ need not be symmetric, but it is similar to the symmetric positive definite $(\hat{G}_t + \delta_t)^{-1/2}\hat{G}_t(\hat{G}_t + \delta_t)^{-1/2}$, so its eigenvalues are real and lie in $(0, 1)$ for any $\delta_t \succ 0$. Eigenvalues in $(0,1)$ are what make $A_t$ a contraction away from OLS. 

\subsection{Ridge balancing weights equals ridge regression}
\label{sec:ridge-equiv}

The ridge Riesz first-order conditions generalise \eqref{eq:eta-forward} to:
\begin{align}\label{eq:ridge-riesz-focs}
    \hat\eta_1^{R} &= (\hat{G}_1 + \delta_1)^{-1}\hat{\mathbb{E}}_0[\phi_1^d], \notag\\
    \hat\eta_t^{R} &= (\hat{G}_t + \delta_t)^{-1}\hat{M}_{t-1}'\hat\eta_{t-1}^{R}, \quad t = 2,\ldots,T,
\end{align}
with $\hat\alpha_t^R = \phi_t'\hat\eta_t^R$. Similarly, the ridge outcome regressions generalise \eqref{eq:beta-backward}:
\begin{align}\label{eq:ridge-outcome-focs}
    \hat\beta_T^{R} &= (\hat{G}_T + \lambda_T)^{-1}\hat{\mathbb{E}}_T[Y\phi_T], \notag\\
    \hat\beta_t^{R} &= (\hat{G}_t + \lambda_t)^{-1}\hat{M}_t\hat\beta_{t+1}^{R}, \quad t = T-1,\ldots,1.
\end{align}

Here, the balancing weights estimator built from the ridge Riesz regression still equals the plug-in regression built from the ridge outcome pass, as long as both use the same penalty matrices, $\delta_t = \lambda_t$ at every stage. The estimators are not equal to the doubly robust estimator anymore, as was the case for OLS. 

\begin{lemma}[Ridge balancing weights equals ridge regression]\label{lem:ridge-equiv}
Let the two recursive regressions share the penalty matrix at every stage: $\delta_t = \lambda_t$ for $t = 1, \ldots, T$. Then, the balancing weights estimator equals the plug-in estimator:
\[
    \hat\theta^{Q,R} \;=\; \hat{\mathbb{E}}_T[Y\phi_T']\hat\eta_T^{R}
    \;=\; \hat{\mathbb{E}}_0[(\phi_1^d)']\hat\beta_1^{R} \;=\; \hat\theta^{P,R}.
\]
\end{lemma}

\begin{proof}
See Appendix~\ref{app:proof-ridge-lemma}.
\end{proof}

This equivalence is specific to ridge with equal penalties. At $T=1$, this equivalence reduces to \cite{kallus2020generalized,hirshberg2019minimaxlinear}.

\subsection{Ridge Riesz with general outcome regressions}
\label{sec:ridge-general}

We now characterise the doubly robust estimator when the Riesz representers are estimated by ridge but the outcome regressions are estimated with a general linear estimator.

\begin{theorem}[Recursive ridge shrinkage]\label{prop:ridge-shrinkage}
Assume each Gram matrix $\hat{G}_t$ is invertible. With ridge Riesz regressions $\hat\eta_1^R, \ldots, \hat\eta_T^R$ as in~\eqref{eq:ridge-riesz-focs}, with penalties $\delta_1, \ldots, \delta_T$, and arbitrary outcome regressions $\hat\beta_1^{Gen}, \ldots, \hat\beta_T^{Gen}$, the doubly robust estimator satisfies $\hat\theta^{DR} = \hat{\mathbb{E}}_0[(\phi_1^d)']\hat\beta_1^{Aug}$, where the augmented coefficients are defined recursively from the last stage:
\begin{align}
    \hat\beta_T^{Aug} &= (I - A_T)\hat\beta_T^{Gen} + A_T\hat\beta_T^{OLS}, \label{eq:augT}\\
    \hat\beta_t^{Aug\text{-}OLS} &:= \hat{G}_t^{-1}\hat{M}_t\hat\beta_{t+1}^{Aug}, \quad t = T-1,\ldots,1, \label{eq:aug-reg}\\
    \hat\beta_t^{Aug} &= (I - A_t)\hat\beta_t^{Gen} + A_t\hat\beta_t^{Aug\text{-}OLS}, \quad t = T-1,\ldots,1. \label{eq:augt}
\end{align}
At each stage, $\hat\beta_t^{Aug\text{-}OLS}$ are the OLS regression coefficients of the stage-$(t+1)$ augmented generated outcomes $(\phi_{t+1}^d)'\hat\beta_{t+1}^{Aug}$ on $\phi_t$.
\end{theorem}

\begin{proof}
See Appendix~\ref{app:proof-ridge-prop}.
\end{proof}

At time period $T$, the augmented coefficient is an affine blend of the penalised regression $\hat\beta_T^{Gen}$ and the final stage $T$ OLS regression $\hat\beta_T^{OLS}$, weighted by $I - A_T$ and $A_T$. This is equivalent to the single-regression result of \citet{bruns2025augmented}. Each earlier stage has the same form but with a different target: it shrinks towards $\hat\beta_t^{Aug\text{-}OLS}$ instead of $\hat\beta_t^{OLS}$. Here, $\hat\beta_t^{Aug\text{-}OLS}$ are the coefficients of an OLS regression of the augmented generated outcomes from the stage before. 

For the high-dimensional regime $k_t > n_t$ where $\hat{G}_t$ is not invertible, there the OLS targets $\hat\beta_t^{OLS}$ and $\hat\beta_t^{Aug\text{-}OLS}$ are read as minimum-norm solutions, with $\hat{G}_t^{-1}$ replaced by the Moore--Penrose pseudoinverse. Appendix~\ref{app:kernel-ridge} shows equivalent results  where the dictionaries are infinite-dimensional kernel features.

\paragraph{Comparison with \citet{bruns2025augmented}.}
For $T = 1$ the recursion collapses to a single equation:
\begin{equation}\label{eq:bruns-smith}
    \hat\beta_1^{Aug} = (I - A_1)\hat\beta_1^{Gen} + A_1\hat\beta_1^{OLS},
    \qquad
    \hat\theta^{DR} = \hat{\mathbb{E}}_0[(\phi_1^d)']\hat\beta_1^{Aug},
\end{equation}
recovering Proposition~3.2 of \citet{bruns2025augmented} exactly. For $T \geq 2$, the last stage $T$ satisfies~\eqref{eq:bruns-smith} directly, while earlier stages have the augmented generated outcome as target, and therefore the final augmented plug-in estimator will have shrunk towards something else than recursive OLS.

\paragraph{Geometric decay of OLS weight.}
The number of time periods determines how much weight on the pure OLS plug-in coefficients is retained. Take the scalar diagonal case where $\hat{G}_t = \sigma_t^2 I$, then $A_t = a_t I$ with $a_t = \sigma_t^2/(\sigma_t^2+\delta_t) \in (0,1)$. Unrolling the recursion gives:
\begin{align*}
    \hat\beta_T^{Aug} &= (1-a_T)\hat\beta_T^{Gen} + a_T\hat\beta_T^{OLS}, \\
    \hat\beta_{T-1}^{Aug} &= (1-a_{T-1})\hat\beta_{T-1}^{Gen} + a_{T-1}(1-a_T)\hat{G}_{T-1}^{-1}\hat{M}_{T-1}\hat\beta_T^{Gen} + a_{T-1}a_T\hat\beta_{T-1}^{OLS},
\end{align*}
and so on until stage 1. Every term but the last contains some regularisation. The one term without the general regression coefficients, which we call the OLS anchor, is $\bigl(\prod_{t=1}^T a_t\bigr)\hat\beta_1^{OLS}$. Note that this recursively shrunk OLS anchor is equal to a ridge plug-in estimator with the penalties from the Riesz regressions, which is $\hat\beta_1^{R}$ under penalties $\delta_t = \lambda_t$. 

\begin{corollary}[Geometric decay of OLS weight]\label{cor:geometric}
Consider the scalar diagonal case $\hat{G}_t = \sigma_t^2 I$ with Riesz shrinkage $a_t = \sigma_t^2/(\sigma_t^2+\delta_t) \in (0,1)$ for $t = 1, \ldots, T$.
\begin{enumerate}[label=(\roman*)]
    \item For arbitrary outcome regressions $\hat\beta_t^{Gen}$, the coefficient on the recursive OLS plug-in coefficients $\hat\beta_1^{OLS}$ is $\prod_{t=1}^T a_t$.
\end{enumerate}
\end{corollary}

\begin{proof}
See Appendix~\ref{app:proof-ridge-cor}.
\end{proof}

\paragraph{Shrinkage in the non-diagonal case.}
For non-diagonal Gram matrices, we derive a contraction in the Euclidean norm. For a scalar penalty $\delta_t I > 0$ the matrix $A_t$ is symmetric, with spectral norm $\rho_t := \|A_t\|_2 = \max_j \mu_{t,j}/(\mu_{t,j}+\delta_t) < 1$, where $\mu_{t,j}$ are eigenvalues of $\hat{G}_t$.

\begin{corollary}[Contraction bound for general Gram matrices]\label{cor:contraction}
Let each $\hat{G}_t$ be invertible and each Riesz penalty a positive scalar, such that $\delta_t I > 0$. For arbitrary linear outcome regressions, the anchor satisfies, in the Euclidean norm,
\[
    \bigl\|\hat\beta_1^{Anc}\bigr\|_2 \;\leq\; \Bigl(\prod_{t=1}^{T}\rho_t\Bigr)\Bigl(\prod_{t=1}^{T-1}\bigl\|\hat{G}_t^{-1}\hat{M}_t\bigr\|_2\Bigr)\bigl\|\hat\beta_T^{OLS}\bigr\|_2,
    \qquad \rho_t = \max_j \frac{\mu_{t,j}}{\mu_{t,j}+\delta_t} \,<\, 1.
\]
\end{corollary}

\begin{proof}
See Appendix~\ref{app:proof-ridge-contraction}.
\end{proof}

For (non-scalar) diagonal penalty matrices, $A_t$ need not be symmetric and its Euclidean norm can exceed one even though its eigenvalues stay in $(0,1)$. However, the contraction then survives stage by stage in the $(\hat{G}_t+\delta_t)$-weighted norms, with $\rho_t$ replaced by the largest eigenvalue of $A_t$. 

\subsection{Double ridge: Ridge Riesz and ridge outcome}
\label{sec:ridge-ridge}

We now consider the doubly robust estimator where the outcome regression is also estimated with ridge penalties, and a shared penalty matrix at every stage, so $\delta_t = \lambda_t$. 

\begin{corollary}[Double ridge as undersmoothed recursive ridge]\label{cor:double-ridge}
Under double ridge ($\hat\beta_t^{Gen} = \hat\beta_t^R$ and $\delta_t = \lambda_t$ at every stage), the doubly robust estimator is $\hat\theta^{DR} = \hat{\mathbb{E}}_0[(\phi_1^d)']\hat\beta_1^{Aug}$ where the augmented coefficients follow the recursion \eqref{eq:augT}--\eqref{eq:augt} with $\hat\beta_t^{Gen} = \hat\beta_t^R$:
\begin{align}
    \hat\beta_T^{Aug} &= (I-A_T)\hat\beta_T^R + A_T\hat\beta_T^{OLS}, \label{eq:augT-ridge}\\
    \hat\beta_t^{Aug} &= (I-A_t)\hat\beta_t^R + A_t\hat\beta_t^{Aug\text{-}OLS}, \quad t = T-1,\ldots,1, \label{eq:augt-ridge}
\end{align}
with $\hat\beta_t^{Aug\text{-}OLS} = \hat{G}_t^{-1}\hat{M}_t\hat\beta_{t+1}^{Aug}$.
\end{corollary}

The final stage ridge outcome coefficients can be written as $\hat\beta_T^R = A_T\hat\beta_T^{OLS}$, and substituting into~\eqref{eq:augT-ridge} gives
\[
    \hat\beta_T^{Aug} = \bigl[I - (I - A_T)^2\bigr]\hat\beta_T^{OLS},
\]
so the augmented coefficients square the shrinkage towards OLS. When $\lambda_T$ and $\hat{G}_T$ commute, in particular for a scalar penalty $\lambda_T = \lambda I$, the squared gap is itself a ridge shrinkage. Like the undersmoothing characterisation of \citet{bruns2025augmented}, $\hat\beta_T^{Aug}$ is a ridge estimator with the reduced effective penalty
\begin{equation}\label{eq:effective-penalty}
    \Gamma_T = \lambda_T(\hat{G}_T + 2\lambda_T)^{-1}\lambda_T \prec \lambda_T,
\end{equation}
which in the scalar case $\hat{G}_T = \sigma_T^2 I$ is $\Gamma_T = \lambda_T^2/(\sigma_T^2 + 2\lambda_T)$. 

Each sequential stage then shrinks the outcome regression coefficients $\hat\beta_t^R$ not towards the OLS plug-in coefficients, but towards an OLS regression on the augmented generated outcome from the previous stage.


\section{General regularisation}
\label{sec:general}

Finally, we consider $P_t$ and $Q_t$ general convex penalties. The stage coefficients have no closed form. We derive a telescoping identity relating the balancing weight estimator to OLS, and a penalty-gradient decomposition relating the doubly robust estimator to OLS. 

\subsection{Setup and the Riesz residual}
\label{sec:general-setup}

Recall the Riesz target, the right-hand side of the forward pass~\eqref{eq:eta-forward}:
\begin{equation}\label{eq:riesz-target}
    \hat\tau_1 := \hat{\mathbb{E}}_0[\phi_1^d], \qquad \hat\tau_t := \hat{M}_{t-1}'\hat\eta_{t-1} \;\text{ for } t \geq 2.
\end{equation}
For any coefficient vectors $\hat\eta_1, \ldots, \hat\eta_T$, define the Riesz residual at stage $t$ as
\begin{equation}\label{eq:riesz-residual}
    c_t \;:=\; \hat\tau_t - \hat{G}_t\hat\eta_t.
\end{equation}
When $\hat\eta_t$ solves the stage-$t$ Riesz regression~\eqref{eq:riesz-t} with convex penalty $Q_t$ and parameter $\delta_t \geq 0$, convex optimality is the subgradient inclusion
\begin{equation}\label{eq:riesz-foc-general}
    c_t \;\in\; \tfrac{\delta_t}{2}\,\partial Q_t(\hat\eta_t),
\end{equation}
so $c_t = \tfrac{\delta_t}{2}\hat{s}_t$ for the subgradient $\hat{s}_t \in \partial Q_t(\hat\eta_t)$ the solution selects, and $c_t = \tfrac{\delta_t}{2}\nabla Q_t(\hat\eta_t)$ when $Q_t$ is differentiable. Under OLS Riesz regressions ($\delta_t = 0$) there is exact balance, $\hat{G}_t\hat\eta_t = \hat\tau_t$, so $c_t = 0$ at every stage. Under ridge we obtain $c_t = \delta_t\hat\eta_t^R$.

\begin{remark}[Cross-stage dependence]\label{rem:coupling}
The residuals are not independent. Instead, at stage $t \geq 2$ the target $\hat\tau_t = \hat{M}_{t-1}'\hat\eta_{t-1}$ contains $\hat\eta_{t-1}$, which itself satisfies $\hat{G}_{t-1}\hat\eta_{t-1} = \hat\tau_{t-1} - c_{t-1}$, so $c_t$ contains previous residuals $c_1, \ldots, c_{t-1}$ through the forward Riesz recursion. 
\end{remark}

\subsection{Telescoping identity for balancing weights}
\label{sec:general-ipw}

The first identity writes the IPW estimator as OLS corrected by the Riesz residuals. At $T = 1$, \citet[Proposition~3.1]{bruns2025augmented} show that any linear balancing weights estimate the same quantity as a regression on a shifted feature profile, $\hat\theta^{Q} = (\hat{\mathbb{E}}_0[\phi_1^d] - c_1)'\hat\beta_1^{OLS}$. The regularisation shifts the target $\hat{\mathbb{E}}_0[\phi_1^d]$ by $c_1$. 

\begin{lemma}[Telescoping identity for balancing weights]\label{prop:telescoping}
Assume each Gram matrix $\hat{G}_t$ is invertible, $t = 1, \ldots, T$. Let $\hat\eta_1, \ldots, \hat\eta_T$ be any coefficient vectors, $\hat\eta_t \in \mathbb{R}^{k_t}$, with targets $\hat\tau_t$ as in~\eqref{eq:riesz-target} and residuals $c_t$ as in~\eqref{eq:riesz-residual}, and let $\hat\beta_t^{OLS}$ denote the OLS outcome regression coefficients~\eqref{eq:backward-ols}. Then
\begin{equation}\label{eq:telescoping}
    \hat\theta^{Q} \;=\; \hat\theta^{OLS} - \sum_{t=1}^T c_t'\hat\beta_t^{OLS}.
\end{equation}
\end{lemma}

\begin{proof}
See Appendix~\ref{app:proof-telescoping}.
\end{proof}

\begin{remark}[Feature-shift representation]
At $T = 1$, equation~\eqref{eq:telescoping} is $\hat\theta^{Q} = (\hat{\mathbb{E}}_0[\phi_1^d] - c_1)'\hat\beta_1^{OLS}$, the feature-shift representation of \citet[Proposition~3.1]{bruns2025augmented}. For $T \geq 2$ the correction $\sum_t c_t'\hat\beta_t^{OLS}$ changes OLS coefficients at every stage, so no single shift of stage-1 covariates reproduces the balancing weights estimator.
\end{remark}

\subsection{The doubly robust estimator under general penalties}
\label{sec:general-dr}

Here, we decompose the doubly robust estimator by combining Lemma \ref{prop:telescoping} with a similar telescoping identity for $\hat\theta^{DR}$ to $\hat\theta^{Q}$. 

\begin{theorem}[Telescoping identity for the doubly robust estimator]\label{prop:general-shrinkage}
Let $\hat\eta_1, \ldots, \hat\eta_T$ be any coefficient vectors, $\hat\eta_t \in \mathbb{R}^{k_t}$, with targets $\hat\tau_t$ as in~\eqref{eq:riesz-target} and residuals $c_t$ as in~\eqref{eq:riesz-residual}, and let $\hat\beta_1^{Gen}, \ldots, \hat\beta_T^{Gen}$ be any outcome regressions. Then,
\begin{equation}\label{eq:dr-ipw-identity}
    \hat\theta^{DR} \;=\; \hat\theta^{Q} + \sum_{t=1}^T c_t'\hat\beta_t^{Gen}.
\end{equation}
If in addition each $\hat{G}_t$ is invertible, with $\hat\beta_t^{OLS}$ the OLS backward-pass coefficients~\eqref{eq:backward-ols}, then
\begin{equation}\label{eq:general-decomposition}
    \hat\theta^{DR} \;=\; \hat\theta^{OLS} + \sum_{t=1}^T c_t'\bigl(\hat\beta_t^{Gen} - \hat\beta_t^{OLS}\bigr).
\end{equation}
\end{theorem}

\begin{proof}
See Appendix~\ref{app:proof-general}.
\end{proof}

This result shows why the doubly robust estimator shrinks less to OLS as the number of time periods increases. 

\begin{remark}[Balancing weights mirror interpretation]
In a comment on \cite{bruns2025augmented}, \cite{liu2025discussion} and \cite{shen2025discussion} note that one can equivalently write the doubly robust estimator as a weighting estimator with augmented weights that undersmooth relative to the base weighting
model. The results in this paper can also be mirrored to the perspective of balancing weights.
\end{remark}


\section{Conclusion}
\label{sec:conclusion}

This paper takes the dissection of \citet{bruns2025augmented} from a single penalised regression into the recursive setting of \citet{chernozhukov2022nested}, in which time-varying treatment effects, surrogate identified treatment effects, and mediation effects use $T$ stages of regression. 

The analysis characterises the resulting estimators under three forms of regularisation. 
Under OLS at every stage the plug-in, balancing weight, and doubly robust estimators coincide exactly.
Under ridge penalties, the weight on the pure-OLS fit decays geometrically in the number of time periods. 
Under a general convex penalty the closed forms break down, but each stage still contributes one correction term, an inner product between the stage Riesz residual and the gap between the penalised and OLS outcome coefficients. 

There are two takeaways. First, the main equivalence from cross sectional causal inference extends to longitudinal causal inference. Second, the interpretation of debiasing as undersmoothing weakens.

\spacingset{1.5}
\bibliography{references}

\newpage
\appendix

\section{Proof for Section~\ref{sec:setup}}\label{sec:setup-proof}

Fix a stage $t$ and condition on $\hat\eta_{t-1}$, which defines $\hat\tau_t$. The stage-$t$ problem trades the variance term $\tfrac{1}{2}\eta'\hat{G}_t\eta$ against the imbalance $\|\hat{G}_t\eta - \hat\tau_t\|_*$ and the penalty $Q_t$, which is the structure of the $T=1$ problem of \citet[Section~2.3, Appendix~A]{bruns2025augmented} with their target $\hat{\mathbb{E}}[\phi^d]$ replaced by the fixed vector $\hat\tau_t$; their equivalence therefore applies stage by stage. 

\emph{Constrained $\Leftrightarrow$ penalised.} Since $\hat{G}_t$ is invertible, $\eta = \hat{G}_t^{-1}\hat\tau_t$ attains zero imbalance and is strictly feasible for any $\delta_t^{(C)} > 0$, so Slater's condition holds and Lagrangian duality makes~\eqref{eq:nested-constrained} and~\eqref{eq:nested-penalised} equivalent under a monotone map $\delta_t^{(C)} \mapsto \delta_t^{(P)}$, with the multiplier on the constraint playing the role of the penalty weight.

\emph{Penalised $\Leftrightarrow$ Riesz.} In the Euclidean case the first-order condition of~\eqref{eq:nested-penalised} is $2\hat{G}_t(\hat{G}_t\eta - \hat\tau_t) + 2\delta_t^{(P)}\hat{G}_t\eta = 0$; cancelling the invertible $\hat{G}_t$ leaves $\hat{G}_t\eta - \hat\tau_t + \delta_t^{(P)}\eta = 0$, which is the first-order condition of~\eqref{eq:nested-riesz-form} with $Q_t(\eta) = \|\eta\|_2^2$ at $\delta_t^{(R)} = \delta_t^{(P)}$. For the $\ell_\infty$ imbalance the conditions become subgradient inclusions and the equivalence requires pairing the norm with its dual penalty, $\ell_\infty$ with $\ell_1$. The bijection $\delta_t^{(P)} \leftrightarrow \delta_t^{(R)}$ is that of \citet[Appendix~A]{bruns2025augmented}, sample-dependent (here also through $\hat\tau_t$) and, for the $\ell_\infty$--$\ell_1$ pair, holding between solution sets. None of the steps uses the form of $\hat\tau_t$, so the proof is equivalent to that of the cross-sectional setting, and the equivalence holds at every stage. \hfill$\square$

\section{Proof for Section~\ref{sec:ols}}
\label{app:proof-ols}

\textit{Part 1: Each correction vanishes.}
The OLS normal equations make each stage's residual orthogonal to that stage's regressors. At the last stage,
\[
    \hat{G}_T\hat\beta_T^{OLS} = \hat{\mathbb{E}}_T[Y\phi_T]
    \quad\Longrightarrow\quad
    \hat{\mathbb{E}}_T[\phi_T\hat\varepsilon_T] = 0,
\]
where $\hat\varepsilon_T = Y - \phi_T'\hat\beta_T^{OLS}$. At each outer stage $t < T$, the equation $\hat{G}_t\hat\beta_t^{OLS} = \hat{M}_t\hat\beta_{t+1}^{OLS}$ gives
\[
    \hat{\mathbb{E}}_t[\phi_t \hat\varepsilon_t] = 0,
    \quad\text{where } \hat\varepsilon_t = (\phi_{t+1}^d)'\hat\beta_{t+1}^{OLS} - \phi_t'\hat\beta_t^{OLS}
\]
is the stage-$t$ residual~\eqref{eq:residuals} at the OLS fits. Since each estimated representer $\hat\alpha_t = \phi_t'\hat\eta_t^{OLS}$ lies in the span of $\phi_t$,
\[
    \hat{\mathbb{E}}_t[\hat\alpha_t\,\hat\varepsilon_t]
    = (\hat\eta_t^{OLS})'\hat{\mathbb{E}}_t[\phi_t\hat\varepsilon_t] = 0,
    \quad t = 1,\ldots,T.
\]
This establishes~\eqref{eq:ols-corrections-zero}, and with~\eqref{eq:dr-estimator} it gives $\hat\theta^{DR} = \hat\theta^{P}$. 

\medskip
\textit{Part 2: All terms are equal.}
We show that $(\hat\eta_t^{OLS})'\hat{G}_t\hat\beta_t^{OLS}$ takes the same value at every stage, and is equal to all the cross-stage inner products $(\hat\eta_t^{OLS})'\hat{M}_t\hat\beta_{t+1}^{OLS}$. Unlike Part 1, this part uses the OLS FOCs of both passes.

Using the last-stage OLS condition $\hat{\mathbb{E}}_T[Y\phi_T] = \hat{G}_T\hat\beta_T^{OLS}$:
\[
    \hat\theta^{Q} = \hat{\mathbb{E}}_T[Y\phi_T']\hat\eta_T^{OLS}
    = (\hat\eta_T^{OLS})'\hat{G}_T\hat\beta_T^{OLS}.
\]
Now step from stage $t+1$ to stage $t$. The forward OLS FOC gives $\hat{M}_t'\hat\eta_t^{OLS} = \hat{G}_{t+1}\hat\eta_{t+1}^{OLS}$, that is $(\hat\eta_t^{OLS})'\hat{M}_t = (\hat\eta_{t+1}^{OLS})'\hat{G}_{t+1}$, and the backward OLS FOC gives $\hat{M}_t\hat\beta_{t+1}^{OLS} = \hat{G}_t\hat\beta_t^{OLS}$. Therefore
\[
    (\hat\eta_{t+1}^{OLS})'\hat{G}_{t+1}\hat\beta_{t+1}^{OLS}
    = (\hat\eta_t^{OLS})'\hat{M}_t\hat\beta_{t+1}^{OLS}
    = (\hat\eta_t^{OLS})'\hat{G}_t\hat\beta_t^{OLS},
\]
where the first equality uses the forward FOC and the second the backward FOC. The cross-stage inner product equals the same-stage at both adjacent stages. Iterating from $t = T$ down to $t = 1$:
\[
    \hat\theta^{Q}
    = (\hat\eta_T^{OLS})'\hat{G}_T\hat\beta_T^{OLS}
    = (\hat\eta_{T-1}^{OLS})'\hat{G}_{T-1}\hat\beta_{T-1}^{OLS}
    = \cdots
    = (\hat\eta_1^{OLS})'\hat{G}_1\hat\beta_1^{OLS}
    = \hat{\mathbb{E}}_0[(\phi_1^d)']\hat\beta_1^{OLS}
    = \hat\theta^{P},
\]
where the last step uses $\hat\eta_1^{OLS} = \hat{G}_1^{-1}\hat{\mathbb{E}}_0[\phi_1^d]$ and the symmetry of $\hat{G}_1$. This establishes both~\eqref{eq:ols-bilinear} and the cross-stage equality. \hfill$\square$
\section{Proofs for Section~\ref{sec:ridge}}
\label{app:proof-ridge}

\subsection{Proof of Lemma~\ref{lem:ridge-equiv}}
\label{app:proof-ridge-lemma}

Under the lemma's hypothesis the two passes share the penalty matrices, $\delta_t = \lambda_t$, and we write $\lambda_t$ for both. Unrolling the forward Riesz recursion~\eqref{eq:ridge-riesz-focs} from stage $T$ down to stage 1:
\[
    \hat\eta_T^R
    = (\hat{G}_T+\lambda_T)^{-1}\hat{M}_{T-1}'(\hat{G}_{T-1}+\lambda_{T-1})^{-1}\hat{M}_{T-2}'
      \cdots(\hat{G}_2+\lambda_2)^{-1}\hat{M}_1'(\hat{G}_1+\lambda_1)^{-1}\hat{\mathbb{E}}_0[\phi_1^d],
\]
a product with one factor $(\hat{G}_t+\lambda_t)^{-1}$ per stage and one cross-moment $\hat{M}_{t-1}'$ between adjacent stages. The balancing estimator  is $\hat\theta^{Q,R} = \hat{\mathbb{E}}_T[Y\phi_T']\hat\eta_T^R$. Also, using that each $(\hat{G}_t+\lambda_t)^{-1}$ is symmetric, and transposing the (scalar) estimator, we get:
\[
    \hat\theta^{Q,R}
    = \hat{\mathbb{E}}_0[(\phi_1^d)'](\hat{G}_1+\lambda_1)^{-1}\hat{M}_1(\hat{G}_2+\lambda_2)^{-1}\hat{M}_2\cdots(\hat{G}_T+\lambda_T)^{-1}\hat{\mathbb{E}}_T[Y\phi_T].
\]
Unrolling the backward outcome recursion~\eqref{eq:ridge-outcome-focs} gives $\hat\beta_1^R = (\hat{G}_1+\lambda_1)^{-1}\hat{M}_1(\hat{G}_2+\lambda_2)^{-1}\hat{M}_2\cdots(\hat{G}_T+\lambda_T)^{-1}\hat{\mathbb{E}}_T[Y\phi_T]$, the same product, so the right-hand side equals $\hat{\mathbb{E}}_0[(\phi_1^d)']\hat\beta_1^R = \hat\theta^{P,R}$. \hfill$\square$

\subsection{Proof of Theorem~\ref{prop:ridge-shrinkage}}
\label{app:proof-ridge-prop}

Each $\hat{G}_t$ is invertible, as assumed in the proposition. We prove a slightly more general statement, indexed by the outer target. For a vector $w$, let the stage-1 Riesz coefficients solve $\hat\eta_1^R = (\hat{G}_1+\delta_1)^{-1}w$, and write $\hat\theta^{DR}(w)$ for the doubly robust estimator whose outer feature profile is $w$, so the proposition is the case $w = \hat{\mathbb{E}}_0[\phi_1^d]$. The claim is $\hat\theta^{DR}(w) = w'\hat\beta_1^{Aug}$. 
\medskip
\textit{Base case ($T=1$).} The doubly robust moment is
\[
    \hat\theta^{DR}(w)
    = w'\hat\beta_1^{Gen}
    + (\hat\eta_1^R)'\bigl(\hat{\mathbb{E}}_1[Y\phi_1] - \hat{G}_1\hat\beta_1^{Gen}\bigr).
\]
Using $(\hat\eta_1^R)' = w'(\hat{G}_1+\delta_1)^{-1} = w'A_1\hat{G}_1^{-1}$ and $\hat\beta_1^{OLS} = \hat{G}_1^{-1}\hat{\mathbb{E}}_1[Y\phi_1]$:
\begin{align*}
    \hat\theta^{DR}(w)
    &= w'(I-A_1)\hat\beta_1^{Gen}
     + w'A_1\hat\beta_1^{OLS}
     = w'\hat\beta_1^{Aug}.
\end{align*}

\medskip
\textit{Inductive step.} Suppose the claim holds for $T-1$ stages and every outer target. The $T$-stage doubly robust estimator with outer target $w$ splits as
\[
    \hat\theta^{DR}(w)
    = w'\hat\beta_1^{Gen}
    + (\hat\eta_1^R)'\bigl(\hat{M}_1\hat\beta_2^{Gen} - \hat{G}_1\hat\beta_1^{Gen}\bigr)
    + \sum_{t=2}^T\hat{\mathbb{E}}_t[\hat\alpha_t^R\hat\varepsilon_t].
\]
The term $(\hat\eta_1^R)'\hat M_1\hat\beta_2^{Gen}$ together with the sum over $t=2,\ldots,T$ is the doubly robust estimator for the inner $(T-1)$-stage problem on stages $2,\ldots,T$, whose outer target is $\hat{M}_1'\hat\eta_1^R$. By the inductive hypothesis at that target,
\[
    (\hat\eta_1^R)'\hat{M}_1\hat\beta_2^{Gen} + \sum_{t=2}^T\hat{\mathbb{E}}_t[\hat\alpha_t^R\hat\varepsilon_t]
    = (\hat\eta_1^R)'\hat{M}_1\hat\beta_2^{Aug},
\]
where $\hat\beta_2^{Aug}$ follows the recursion~\eqref{eq:augT}--\eqref{eq:augt} on stages $2,\ldots,T$. Substituting and using $(\hat\eta_1^R)'\hat{G}_1 = w'A_1$ together with $(\hat\eta_1^R)'\hat{M}_1 = w'A_1\hat{G}_1^{-1}\hat{M}_1$:
\begin{align*}
    \hat\theta^{DR}(w)
    &= w'(I-A_1)\hat\beta_1^{Gen}
     + w'A_1\underbrace{\hat{G}_1^{-1}\hat{M}_1\hat\beta_2^{Aug}}_{=\,\hat\beta_1^{Aug\text{-}OLS}}
     = w'\hat\beta_1^{Aug}.
\end{align*}
Setting $w = \hat{\mathbb{E}}_0[\phi_1^d]$ gives the proposition. \hfill$\square$

\subsection{Proof of Corollary~\ref{cor:geometric}}
\label{app:proof-ridge-cor}

Throughout, $\hat{G}_t = \sigma_t^2 I$, so $A_t = a_t I$ with $a_t = \sigma_t^2/(\sigma_t^2+\delta_t)$ and $(\hat{G}_t+\delta_t)^{-1} = a_t\hat{G}_t^{-1}$.

\medskip
\textit{Part (i).} Substitute~\eqref{eq:aug-reg} into~\eqref{eq:augt} and iterate from stage 1 inward. Every branch of the recursion that takes a factor $(1-a_t)$ terminates in some $\hat\beta_t^{Gen}$. The single branch that takes the factor $a_t$ at every stage terminates in $\hat\beta_T^{OLS}$ and contributes
\[
    a_1\hat{G}_1^{-1}\hat{M}_1\, a_2\hat{G}_2^{-1}\hat{M}_2 \cdots a_{T-1}\hat{G}_{T-1}^{-1}\hat{M}_{T-1}\, a_T\hat{G}_T^{-1}\hat{\mathbb{E}}_T[Y\phi_T]
    = \Bigl(\textstyle\prod_{t=1}^T a_t\Bigr)\hat\beta_1^{OLS},
\]
by the OLS backward pass~\eqref{eq:backward-ols}.

\subsection{Proof of Corollary~\ref{cor:contraction}}
\label{app:proof-ridge-contraction}

Part~(i) of the preceding proof did not use the scalar structure. Instead, substituting~\eqref{eq:aug-reg} into~\eqref{eq:augt} and iterating from stage 1 inward, every part of the recursion that takes a factor $I - A_t$ terminates in some $\hat\beta_t^{Gen}$, and the single part that is weighted by $A_t$ at every stage ends in $\hat\beta_T^{OLS}$.

For a scalar penalty $\delta_t > 0$, the matrix $A_t = (\hat{G}_t + \delta_t I)^{-1}\hat{G}_t$ is a product of two symmetric matrices that commute, both functions of $\hat{G}_t$, hence symmetric, with eigenvalues $\mu_{t,j}/(\mu_{t,j} + \delta_t)$. The map $\mu \mapsto \mu/(\mu + \delta_t)$ is increasing and below one, so $\|A_t\|_2 = \max_j \mu_{t,j}/(\mu_{t,j} + \delta_t) = \rho_t < 1$. The bound is submultiplicativity of the Euclidean operator norm along the product of the anchor:
\[
    \bigl\|\hat\beta_1^{Anc}\bigr\|_2
    \;\leq\; \|A_1\|_2\,\bigl\|\hat{G}_1^{-1}\hat{M}_1\bigr\|_2\,\|A_2\|_2 \cdots \bigl\|\hat{G}_{T-1}^{-1}\hat{M}_{T-1}\bigr\|_2\,\|A_T\|_2\,\bigl\|\hat\beta_T^{OLS}\bigr\|_2,
\]
which is the stated product after collecting the $\rho_t$. \hfill$\square$
\section{Proofs for Section~\ref{sec:general}}
\label{app:proof-general-all}

\subsection{Proof of Lemma~\ref{prop:telescoping}}
\label{app:proof-telescoping}

For $t = 1, \ldots, T$, define the bilinear forms
\[
    S_t := (\hat\beta_t^{OLS})'\hat{G}_t\hat\eta_t.
\]
The proof has three steps. (1) $S_T$ is the IPW estimator, (2) the $S_t$ obey a one-step recursion that peels off one correction per stage, and (3) and the recursion bottoms out at $\hat\theta^{OLS}$.

\emph{Step 1.} By the last-stage OLS normal equation $\hat{G}_T\hat\beta_T^{OLS} = \hat{\mathbb{E}}_T[Y\phi_T]$ and the symmetry of $\hat{G}_T$,
\[
    \hat\theta^{Q} = \hat{\mathbb{E}}_T[Y\phi_T']\hat\eta_T = (\hat\beta_T^{OLS})'\hat{G}_T\hat\eta_T = S_T.
\]

\emph{Step 2.} Fix $t \geq 2$. The residual definition~\eqref{eq:riesz-residual} gives $\hat{G}_t\hat\eta_t = \hat\tau_t - c_t$ with $\hat\tau_t = \hat{M}_{t-1}'\hat\eta_{t-1}$, so
\[
    S_t = (\hat\beta_t^{OLS})'\hat{M}_{t-1}'\hat\eta_{t-1} - c_t'\hat\beta_t^{OLS}.
\]
The OLS backward pass~\eqref{eq:backward-ols} gives $\hat{M}_{t-1}\hat\beta_t^{OLS} = \hat{G}_{t-1}\hat\beta_{t-1}^{OLS}$, so the first term is $(\hat\beta_{t-1}^{OLS})'\hat{G}_{t-1}\hat\eta_{t-1} = S_{t-1}$ and
\[
    S_t = S_{t-1} - c_t'\hat\beta_t^{OLS}.
\]

\emph{Step 3.} At stage 1 the target is $\hat\tau_1 = \hat{\mathbb{E}}_0[\phi_1^d]$, so the same substitution gives
\[
    S_1 = (\hat\beta_1^{OLS})'\bigl(\hat{\mathbb{E}}_0[\phi_1^d] - c_1\bigr) = \hat\theta^{OLS} - c_1'\hat\beta_1^{OLS}.
\]
Iterating Step 2 from $t = T$ down to $t = 2$ and inserting Step 3,
\[
    \hat\theta^{Q} = S_T = S_1 - \sum_{t=2}^T c_t'\hat\beta_t^{OLS} = \hat\theta^{OLS} - \sum_{t=1}^T c_t'\hat\beta_t^{OLS}.
\]
Beyond the OLS normal equations of the backward pass, the argument uses only the definition of the residuals; no optimality condition on $\hat\eta_t$ enters. \hfill$\square$

\subsection{Proof of Theorem~\ref{prop:general-shrinkage}}
\label{app:proof-general}

Identity~\eqref{eq:general-decomposition} follows from combining~\eqref{eq:dr-ipw-identity} with the telescoping identity~\eqref{eq:telescoping}. It remains to prove~\eqref{eq:dr-ipw-identity}.

By definition of the DR estimator:
\[
    \hat\theta^{DR} = \hat\theta^{P,Gen} + \sum_{t=1}^{T} \hat\eta_t'\hat{\mathbb{E}}_t[\phi_t\hat\varepsilon_t^{Gen}],
\]
where $\hat\theta^{P,Gen} = \hat{\mathbb{E}}_0[(\phi_1^d)']\hat\beta_1^{Gen}$, as defined in Section~\ref{sec:nested-regression}, and $\hat\varepsilon_t^{Gen}$ are the residuals from the general outcome regressions. Expanding the residuals using $\hat\varepsilon_T^{Gen} = Y - \phi_T'\hat\beta_T^{Gen}$ and $\hat\varepsilon_t^{Gen} = (\phi_{t+1}^d)'\hat\beta_{t+1}^{Gen} - \phi_t'\hat\beta_t^{Gen}$ for $t < T$:
\begin{equation}\label{eq:correction-expanded}
    \sum_{t=1}^{T} \hat\eta_t'\hat{\mathbb{E}}_t[\phi_t\hat\varepsilon_t^{Gen}]
    = \underbrace{\hat\eta_T'\hat{\mathbb{E}}_T[Y\phi_T]}_{=\,\hat\theta^{Q}}
    - \sum_{t=1}^{T} \hat\eta_t'\hat{G}_t\hat\beta_t^{Gen}
    + \sum_{t=1}^{T-1}\hat\eta_t'\hat{M}_t\hat\beta_{t+1}^{Gen}.
\end{equation}
The key step rewrites the cross-moment terms through the residuals. The residual definition~\eqref{eq:riesz-residual} at stage $t+1$, whose target is $\hat\tau_{t+1} = \hat{M}_t'\hat\eta_t$, gives $\hat{G}_{t+1}\hat\eta_{t+1} = \hat{M}_t'\hat\eta_t - c_{t+1}$, so
\begin{equation}\label{eq:cross-moment-substitution}
    \hat\eta_t'\hat{M}_t = \hat\eta_{t+1}'\hat{G}_{t+1} + c_{t+1}'.
\end{equation}
Substituting~\eqref{eq:cross-moment-substitution} into~\eqref{eq:correction-expanded} and re-indexing ($t+1 \to t$ in the sum):
\[
    \sum_{t=1}^{T-1}\hat\eta_t'\hat{M}_t\hat\beta_{t+1}^{Gen}
    = \sum_{t=2}^{T}\hat\eta_t'\hat{G}_t\hat\beta_t^{Gen} + \sum_{t=2}^{T}c_t'\hat\beta_t^{Gen}.
\]
The $\hat\eta_t'\hat{G}_t\hat\beta_t^{Gen}$ terms for $t = 2, \ldots, T$ cancel between the two sums in~\eqref{eq:correction-expanded}, leaving:
\[
    \sum_{t=1}^{T} \hat\eta_t'\hat{\mathbb{E}}_t[\phi_t\hat\varepsilon_t^{Gen}]
    = \hat\theta^{Q} - \hat\eta_1'\hat{G}_1\hat\beta_1^{Gen} + \sum_{t=2}^{T}c_t'\hat\beta_t^{Gen}.
\]
The residual definition at stage 1, with target $\hat\tau_1 = \hat{\mathbb{E}}_0[\phi_1^d]$, gives $\hat\eta_1'\hat{G}_1 = \hat{\mathbb{E}}_0[(\phi_1^d)'] - c_1'$, so $\hat\eta_1'\hat{G}_1\hat\beta_1^{Gen} = \hat\theta^{P,Gen} - c_1'\hat\beta_1^{Gen}$. Therefore:
\[
    \hat\theta^{DR} = \hat\theta^{P,Gen} + \hat\theta^{Q} - \hat\theta^{P,Gen} + \sum_{t=1}^{T}c_t'\hat\beta_t^{Gen}
    = \hat\theta^{Q} + \sum_{t=1}^{T}c_t'\hat\beta_t^{Gen}.
\]
Only the definitions of the residuals and targets enter: no optimality condition on $\hat\eta_t$, no invertibility of any $\hat{G}_t$, and no structure on the penalties. \hfill$\square$

\section{Extension to kernel ridge regression}
\label{app:kernel-ridge}

The results of Section~\ref{sec:ridge} are stated for finite-dimensional dictionaries $\phi_t \in \mathbb{R}^{k_t}$ with $k_t \leq n$. This appendix shows that all results extend to the kernel ridge regression (KRR) setting, where the feature maps are induced by reproducing kernel Hilbert spaces (RKHSs) and the dimension may be infinite. The extension parallels \citet[Appendix~B]{bruns2025augmented}, who show that the cross-sectional ridge results apply to the KRR setting. For clarity of exposition, we suppose there is only one measure and $n$ total observations.

\subsection{Setup}
\label{app:kernel-setup}

At each stage $t$, let $\mathcal{H}_t$ be an RKHS on $\mathcal{Z}_t$ with kernel $K_t$. The feature map is $\phi_t(z) = K_t(z, \cdot) \in \mathcal{H}_t$, which may be infinite-dimensional. All estimators below depend on the data only through kernel evaluations at the observed and counterfactual points. We write the leading case of Section~\ref{sec:linear-setup}, in which each $m_t$ substitutes a counterfactual point $Z_t^d$. For a contrast, the counterfactual entries below become the corresponding signed combinations of kernel evaluations, and every identity is linear in those entries. Define the $n \times n$ kernel matrices:
\[
    [K_{p,t}]_{ij} := K_t(Z_{t,i}, Z_{t,j}), \qquad
    [\bar{K}_{q,t}]_i := \hat{\mathbb{E}}\!\left[K_t(Z_t^d, Z_{t,i})\right],
\]
and for the cross-stage kernels ($t < T$), built from the inner stage-$(t+1)$ kernel evaluated at the counterfactual:
\[
    [K_{M,t}]_{ij} := K_{t+1}(Z_{t+1,i}^d, Z_{t+1,j}),
\]
so that $K_{M,t}\beta_{t+1}^{\mathcal{H}}$ is the vector of generated outcomes $\hat{f}_{t+1}(Z_{t+1,i}^d)$ that the stage-$t$ regression takes as target. The stage-$t$ outcome loss evaluates the candidate function only at the observed points $Z_{t,i}$, so by the representer theorem the kernel ridge solution lies in the span of the $n$ observed kernel evaluations: $\hat{f}_t(z) = K_{z,p_t}\beta_t^{\mathcal{H}}$ for a coefficient vector $\beta_t^{\mathcal{H}} \in \mathbb{R}^n$, where $K_{z,p_t}$ is the row vector of kernel evaluations between $z$ and the training points. The Riesz loss evaluates $\alpha$ at the counterfactual points $Z_{t,i}^d$ as well as at the observed ones, so its minimiser lies in the span of both sets of kernel functions, and no expansion over the training points alone represents it. The balancing weights and every debiasing correction read $\hat\alpha_t$ only through its values at the training points, so we track those values directly and define
\[
    [\eta_t^{\mathcal{H}}]_i := \tfrac{1}{n}\,\hat\alpha_t(Z_{t,i}),
\]
the weight on unit $i$, scaled so that the IPW estimator is the plain inner product $\hat\theta^{Q,\mathcal{H}} = Y'\eta_T^{\mathcal{H}}$.

\subsection{Kernel ridge First-Order Conditions}

Here, we set the Riesz penalties equal to the outcome penalties, $\delta_t = \lambda_t$, and write $\lambda_t$ for both. The backward and forward passes~\eqref{eq:ridge-outcome-focs}--\eqref{eq:ridge-riesz-focs} become:
\begin{align}
    \beta_T^{\mathcal{H}} &= (K_{p,T} + \lambda_T)^{-1}Y, &
    \eta_1^{\mathcal{H}} &= (K_{p,1} + \lambda_1)^{-1}\bar{K}_{q,1}, \notag\\
    \beta_t^{\mathcal{H}} &= (K_{p,t} + \lambda_t)^{-1}K_{M,t}\beta_{t+1}^{\mathcal{H}}, &
    \eta_t^{\mathcal{H}} &= (K_{p,t} + \lambda_t)^{-1}K_{M,t-1}'\eta_{t-1}^{\mathcal{H}},
    \label{eq:kernel-focs}
\end{align}
for $t = T{-}1, \ldots, 1$ (backward) and $t = 2, \ldots, T$ (forward). These are identical in structure to~\eqref{eq:ridge-outcome-focs}--\eqref{eq:ridge-riesz-focs}, with $\hat{G}_t \mapsto K_{p,t}$, $\hat{M}_t \mapsto K_{M,t}$, $\hat{\mathbb{E}}[Y\phi_T] \mapsto Y$, and $\hat{\mathbb{E}}[\phi_1^d] \mapsto \bar{K}_{q,1}$.\footnote{Because $K_{p,t} = \Phi_t\Phi_t'$ is unnormalised whereas $\hat{G}_t = n^{-1}\Phi_t'\Phi_t$ contains the $1/n$, the penalty $\lambda_t$ in this appendix equals $n$ times the stage-$t$ penalty of Section~\ref{sec:ridge}. The fitted outcome functions and the weight values coincide under this rescaling.} The push-through identity holds for feature operators as for matrices, so the results apply when $\mathcal{H}_t$ is infinite-dimensional.

\subsection{Extension of the main results}

The results from Section~\ref{sec:ridge} apply to the RKHS setting by making the substitutions above. For the cross-moment, $\hat{M}_t$ becomes $K_{M,t}$, the inner-kernel matrix that turns stage-$(t+1)$ coefficients into the generated outcomes used as the response of the stage-$t$ regression. Given that substitution, the proofs in Appendix~\ref{app:proof-ridge} operate on the matrix product
\[
    (\hat{G}_1 + \lambda_1)^{-1}\hat{M}_1(\hat{G}_2 + \lambda_2)^{-1}\hat{M}_2 \cdots (\hat{G}_T + \lambda_T)^{-1},
\]
and the remaining properties they use are: (i)~each $(\hat{G}_t + \lambda_t)^{-1}$ is symmetric, and (ii)~the product is sandwiched between two vectors to produce a scalar. Both hold equally when $\hat{G}_t$ is replaced by $K_{p,t}$ (symmetric and positive semidefinite) and $\lambda_t \succ 0$ ensures invertibility.

\paragraph{Lemma~\ref{lem:ridge-equiv} (KRR version).} The kernel IPW and kernel regression estimators coincide:
\[
    \hat\theta^{Q,\mathcal{H}} = Y'(K_{p,T}+\lambda_T)^{-1}K_{M,T-1}'\cdots(K_{p,1}+\lambda_1)^{-1}\bar{K}_{q,1}
    = \bar{K}_{q,1}'(K_{p,1}+\lambda_1)^{-1}K_{M,1}\cdots(K_{p,T}+\lambda_T)^{-1}Y = \hat\theta^{P,\mathcal{H}},
\]
exactly as in the proof of Lemma~\ref{lem:ridge-equiv}.

\paragraph{Theorem~\ref{prop:ridge-shrinkage} (KRR version).} Define $A_t^{\mathcal{H}} := (K_{p,t} + \lambda_t)^{-1}K_{p,t}$. The doubly robust estimator satisfies $\hat\theta^{DR,\mathcal{H}} = \bar{K}_{q,1}'\beta_1^{Aug,\mathcal{H}}$, where the augmented kernel coefficients follow the recursion:
\begin{align*}
    \beta_T^{Aug,\mathcal{H}} &= (I - A_T^{\mathcal{H}})\beta_T^{Gen,\mathcal{H}} + A_T^{\mathcal{H}}\beta_T^{OLS,\mathcal{H}}, \\
    \beta_t^{Aug\text{-}OLS,\mathcal{H}} &:= K_{p,t}^{\dagger}K_{M,t}\beta_{t+1}^{Aug,\mathcal{H}}, \\
    \beta_t^{Aug,\mathcal{H}} &= (I - A_t^{\mathcal{H}})\beta_t^{Gen,\mathcal{H}} + A_t^{\mathcal{H}}\beta_t^{Aug\text{-}OLS,\mathcal{H}},
\end{align*}
where $K_{p,t}^{\dagger}$ is the Moore--Penrose pseudoinverse, $\beta_T^{OLS,\mathcal{H}} := K_{p,T}^{\dagger}Y$ is the minimum-norm KRR solution at $\lambda_T = 0$, and $\beta_t^{Aug\text{-}OLS,\mathcal{H}}$ is its earlier-stage analogue: the minimum-norm regression of the augmented generated outcomes $K_{M,t}\beta_{t+1}^{Aug,\mathcal{H}}$ on the stage-$t$ kernel features. The proof follows that of Theorem~\ref{prop:ridge-shrinkage}, with $\hat{G}_t^{-1}$ replaced by $K_{p,t}^{\dagger}$, and one addition for the rank-deficient case. When $K_{p,t}$ is singular the finite collapse $A_t\hat{G}_t^{-1} = (\hat{G}_t+\lambda_t)^{-1}$ becomes
\[
    A_t^{\mathcal{H}}K_{p,t}^{\dagger} = (K_{p,t}+\lambda_t)^{-1}K_{p,t}K_{p,t}^{\dagger} = (K_{p,t}+\lambda_t)^{-1}P_t,
\]
where $P_t := K_{p,t}K_{p,t}^{\dagger}$ projects onto $\mathrm{range}(K_{p,t})$, so $\beta_t^{Aug,\mathcal{H}}$ is pinned only up to $\ker(K_{p,t})$. This does not reach the estimator, by the reproducing property. A coefficient vector $v \in \ker(K_{p,t})$ represents the zero function of $\mathcal{H}_t$, since $\|\sum_j v_j K_t(\cdot, Z_{t,j})\|_{\mathcal{H}_t}^2 = v'K_{p,t}v = 0$, such that $K_{M,t-1}v = 0$ and $\bar{K}_{q,1}'v = 0$. The same argument places each stage target ($\bar{K}_{q,1}$ at stage 1, $K_{M,t-1}'\eta_{t-1}^{\mathcal{H}}$ at later stages) in $\mathrm{range}(K_{p,t})$, which $(K_{p,t}+\lambda_t)^{-1}$ preserves. Every contraction in the proof thus contains only the $P_t$-projected part of $\beta_t^{Aug,\mathcal{H}}$, the null-space slack drops out, and the finite-sample identity applies. For a strictly positive-definite kernel with distinct points $K_{p,t}$ is invertible, $P_t = I$, and the proof is identical.

\paragraph{Corollaries~\ref{cor:geometric} and~\ref{cor:contraction} (KRR version).} In the RKHS setting, the eigenvalues of $A_t^{\mathcal{H}}$ are $\mu_{t,j}/(\mu_{t,j} + \lambda_t)$, where $\mu_{t,j}$ are the eigenvalues of $K_{p,t}$. In the the kernel recursion, the anchor term replaces each $K_{p,t}^{\dagger}$ of the minimum-norm product by $(K_{p,t}+\lambda_t)^{-1}P_t = A_t^{\mathcal{H}}K_{p,t}^{\dagger}$, inserting one contraction per stage. The anchor contributes the recursive KRR plug-in $\hat\theta^{P,\mathcal{H}}$, the kernel analogue of the anchor $\hat\beta_1^R$ in the finite-dimensional recursion. The stage operators do not share eigenvectors in general, so the insertions do not collapse to a single scalar weight as in the scalar diagonal case. Given that the kernel penalties are scalars, each $A_t^{\mathcal{H}}$ is symmetric with spectral norm $\rho_t^{\mathcal{H}} := \max_j \mu_{t,j}/(\mu_{t,j} + \lambda_t) < 1$, and the contraction bound of Corollary~\ref{cor:contraction} carries over with $\hat{G}_t^{-1}\hat{M}_t \mapsto K_{p,t}^{\dagger}K_{M,t}$, $\hat\beta_T^{OLS} \mapsto K_{p,T}^{\dagger}Y$ and $\rho_t \mapsto \rho_t^{\mathcal{H}}$. The anchor's Euclidean norm is bounded by $\prod_{t=1}^T \rho_t^{\mathcal{H}}$ times the submultiplicative bound on the minimum-norm product. The geometric decay of the anchor weight in the number of time periods therefore persists in the kernel setting as it does for general Gram matrices: as a contraction bound, not an exact identity.

\begin{remark}[RKHS equivalence of \cite{bruns2025augmented}]
The extension to KRR changes the matrices but not the algebra. This parallels the observation of \citet[Appendix~B.4]{bruns2025augmented} that their Propositions~3.1 and~3.2 extend to the RKHS setting without alteration.
\end{remark}

\section{Interpretation and cross-stage structure}
\label{sec:general-interpretation}

The decomposition~\eqref{eq:general-decomposition} presents the doubly robust estimator as OLS plus one correction per time period. Each correction is an inner product of the Riesz residual $c_t$ and the gap $\hat\beta_t^{Gen} - \hat\beta_t^{OLS}$ between the penalised outcome regression and OLS.

A stage contributes only when both factors are non-zero. Estimate the stage-$t$ Riesz regression by OLS ($\delta_t = 0$) and $c_t = 0$.  Estimate the outcome regression by OLS ($\hat\beta_t^{Gen} = \hat\beta_t^{OLS}$) and it drops out too. 

\paragraph{Cross-stage coupling.}
The decomposition~\eqref{eq:general-decomposition} is exact in $(c_1, \ldots, c_T)$, but these feature shifts $c_t$ are not independent. Take $T = 2$, the number of time periods of the time-varying treatment, surrogacy and, mediation examples. The two residuals are:
\begin{align*}
    c_1 &= \hat{\mathbb{E}}_0[\phi_1^d] - \hat{G}_1\hat\eta_1, \\
    c_2 &= \hat{M}_1'\hat\eta_1 - \hat{G}_2\hat\eta_2.
\end{align*}
The stage-2 target is $\hat\tau_2 = \hat{M}_1'\hat\eta_1$, and solving the stage-1 residual definition for the coefficients gives $\hat\eta_1 = \hat{G}_1^{-1}(\hat{\mathbb{E}}_0[\phi_1^d] - c_1)$, so
\[
    \hat\tau_2 = \hat{M}_1'\hat{G}_1^{-1}(\hat{\mathbb{E}}_0[\phi_1^d] - c_1)
    = \hat{M}_1'\hat\eta_1^{OLS} - \hat{M}_1'\hat{G}_1^{-1}c_1,
\]
with $\hat\eta_1^{OLS} = \hat{G}_1^{-1}\hat{\mathbb{E}}_0[\phi_1^d]$ the unpenalised first-stage representer. A penalty at stage~1 shifts the stage-2 target by $-\hat{M}_1'\hat{G}_1^{-1}c_1$, and the shift passes into $\hat\eta_2$ and so into $c_2$. The stage-2 residual is not a function of the stage-2 penalty alone, but also reflects the regularisation bias of stage~1.

Specialising to ridge at both stages makes it easier to understand. With the shifted target, the stage-2 Riesz solution is $\hat\eta_2 = (\hat{G}_2 + \delta_2)^{-1}\hat\tau_2 = (\hat{G}_2 + \delta_2)^{-1}(\hat\tau_2^{OLS} - \hat{M}_1'\hat{G}_1^{-1}c_1)$, with $\hat\tau_2^{OLS} := \hat{M}_1'\hat\eta_1^{OLS}$ the stage-2 target under unpenalised stage-1 Riesz. The stage-2 residual $c_2 = \delta_2\hat\eta_2$ then splits in two:
\begin{equation}\label{eq:c2-decomposition}
    c_2 = \underbrace{(I - A_2)'\,\hat\tau_2^{OLS}}_{c_2^{(0)}} \;-\; \underbrace{(I - A_2)'\,\hat{M}_1'\hat{G}_1^{-1}c_1}_{c_2^{(1)}},
\end{equation}
since $\delta_2(\hat{G}_2 + \delta_2)^{-1} = (I - A_2)'$, the transpose of the Section~\ref{sec:ridge} shrinkage gap. Assume $\delta_2$ and $\hat{G}_2$ commute, in particular for a scalar penalty $\delta_2 = \delta I$. The first term $c_2^{(0)}$ is the stage-2 bias that would arise even with OLS at stage~1. The second term, $c_2^{(1)}$, can shrink the target further away from OLS. The extra bias that stage-1 penalisation adds through the cross-moment $\hat{M}_1'$, contracted by $(I - A_2)'$.

Substituting~\eqref{eq:c2-decomposition} into the decomposition~\eqref{eq:general-decomposition} at $T = 2$ splits $\hat\theta^{DR} - \hat\theta^{OLS}$ by stage:
\begin{align*}
    \text{Stage 1:} &\quad c_1'(\hat\beta_1^{Gen} - \hat\beta_1^{OLS}), \\
    \text{Stage 2:} &\quad \bigl[c_2^{(0)} - c_2^{(1)}\bigr]'(\hat\beta_2^{Gen} - \hat\beta_2^{OLS}).
\end{align*}
The inherited piece $c_2^{(1)}$ amplifies or dampens the stage-2 correction according to how $c_1$ aligns with the outcome coefficients, through $\hat{M}_1'\hat{G}_1^{-1}$. In the scalar diagonal case $\hat{G}_t = \sigma_t^2 I$ it collapses to $c_2^{(1)} = (1 - a_2)\hat{M}_1'\sigma_1^{-2}c_1$, and since $c_1 = (1-a_1)\hat{\mathbb{E}}_0[\phi_1^d]$ its size scales as $(1-a_1)(1-a_2)$, the product of the two stage-wise shrinkage gaps. A stage-1 penalty moves the stage-2 correction only at second order in the shrinkage.

The same mechanism holds at any number of time periods. Writing $\hat\eta_t = \hat{G}_t^{-1}(\hat\tau_t - c_t)$ and feeding it into $\hat\tau_{t+1} = \hat{M}_t'\hat\eta_t$ gives the target recursion:
\begin{equation}\label{eq:target-recursion}
    \hat\tau_{t+1} = \hat{M}_t'\hat{G}_t^{-1}\hat\tau_t - \hat{M}_t'\hat{G}_t^{-1}c_t, \quad t = 1, \ldots, T-1.
\end{equation}
Each stage's target is perturbed by the residuals accumulated from all earlier stages, carried forward by $\hat{M}_t'\hat{G}_t^{-1}$. Expanding~\eqref{eq:general-decomposition} all the way down to the penalty parameters $\delta_t$, the Gram matrices and the cross-moments would surface products $\hat{M}_t'\hat{G}_t^{-1}c_t$ at every step. The $c_t$ notation suppresses these products in the display, but they are still present.

\paragraph{Comparison with \citet{bruns2025augmented}.}
At $T = 1$, Theorem~\ref{prop:general-shrinkage} is $\hat\theta^{DR} = \hat\theta^{OLS} + c_1'(\hat\beta_1^{Gen} - \hat\beta_1^{OLS})$. \citet{bruns2025augmented} write the same result (their Proposition~3.2) as a single outcome regression, $\hat\theta^{DR} = \hat{\mathbb{E}}_0[(\phi_1^d)']\hat\beta_1^{Aug}$, with $\hat\beta_1^{Aug}$ an affine blend of $\hat\beta_1^{Gen}$ and $\hat\beta_1^{OLS}$ whose mixing weights encode $c_1$. With recursion there is no such single-coefficient reading. For $T \geq 2$ the correction is $T$ separate inner products, one per stage, and no one augmented coefficient absorbs them. Only with ridge penalisation we can get an analytical solution, where $c_t = \delta_t\hat\eta_t^R$ and a closed-form Riesz recursion shows that the stage corrections fold back into a single backward recursion on augmented coefficients (Theorem~\ref{prop:ridge-shrinkage}).

\paragraph{Relationship to the mixed-bias property.}
The decomposition~\eqref{eq:general-decomposition} has the shape of the population mixed-bias property~\eqref{eq:mixed-bias}: a plain sum over stages, with no explicit cross-stage terms. In the population case the representer errors $\tilde\alpha_t = \alpha_t^* - \alpha_t$ are themselves recursively coupled through the Riesz representation theorem, exactly as the $c_t$ are coupled here. The finite-sample decomposition is the algebraic counterpart of the mixed-bias property, with the Riesz residual $c_t$ in the role of the representer error $\tilde\alpha_t$.

\paragraph{Proximity of the doubly robust estimator to OLS.}

In the a single time period with double ridge regularization and a scalar diagonal weighting matrix, the doubly robust estimator's deviation from OLS is a fixed fraction of the balancing weight estimator's deviation, so the doubly robust estimator is closer to OLS than the IPW or plug-in. Extending this to multiple time periods under the same simplifying assumptions, that changes. The doubly robust remains strictly closer to OLS than the IPW and plug-in, but as more periods are added, the doubly robust estimator moves towards the IPW and plug-in, not towards OLS. The same total weight on the OLS regression that governs the geometric-decay corollary also governs how much closer doubly robust sits to OLS here. 

This whole ordering is not guaranteed in general, as it depends on every direction in the regression shrinking at the same rate. Once the weighting matrix is not diagonal, different directions can contract at different rates, and uneven shrinkage can pull them apart instead of together, potentially leaving doubly robust farther from OLS than the IPW or plug-in.

\end{document}